\documentclass[10pt,journal]{IEEEtran} 

\title{Efficient Multi-Robot Exploration with Energy Constraint based on Optimal Transport Theory }

\author{Rabiul Hasan Kabir and Kooktae Lee
\thanks{R. H. Kabir and K. Lee are with the Department of Mechanical Engineering, New Mexico Institute of Mining and Technology, Socorro, NM 87801, USA.
         {\tt\small rabiul.kabir@student.nmt.edu, kooktae.lee@nmt.edu}}
}

\usepackage{amsmath,graphicx,amsfonts,amssymb,epsfig,subfig,mathrsfs,mathtools}
\usepackage{cuted}
\usepackage{arydshln}
\usepackage{blkarray}

\usepackage{color}
\usepackage{multirow}
\usepackage{multicol}
\usepackage{lipsum}
\usepackage{rotating}
\usepackage{graphicx}%
\usepackage{algorithm}
\usepackage{algpseudocode}
\usepackage{cite}
\usepackage{tikz}
\usepackage{cancel}
\usepackage{amsthm}
\usepackage{textcomp}

\DeclareMathOperator*{\minimize}{minimize}

\DeclareMathOperator*{\argmin}{arg\,min}

\newtheorem{assumption}{Assumption}
\newtheorem{theorem}{Theorem}

\newtheorem{proposition}{Proposition}

\allowdisplaybreaks

\newlength\myindent
\begin{document}
\maketitle

\begin{abstract}
This paper addresses an Optimal Transport (OT)-based efficient multi-robot exploration problem, considering the energy constraints of a multi-robot system. The efficiency in this problem implies how a team of robots (agents) covers a given domain, reflecting a priority of areas of interest represented by a density distribution, rather than simply following a preset of uniform patterns. To achieve an efficient multi-robot exploration, the optimal transport theory that quantifies a distance between two density distributions is employed as a tool, which also serves as a means of performance measure. The energy constraints for the multi-robot system is then incorporated into the OT-based multi-robot exploration scheme. 

The proposed scheme is decoupled from robot dynamics, broadening the applicability of the multi-robot exploration plan to heterogeneous robot platforms.
Not only the centralized but also decentralized algorithms are provided to cope with more realistic scenarios such as communication range limits between agents. 
To measure the exploration efficiency, the upper bound of the performance is developed for both the centralized and decentralized cases based on the optimal transport theory, which is computationally tractable as well as efficient.
The proposed multi-robot exploration scheme is also applicable to a time-varying distribution, where the spatio-temporal evolution of the given reference distribution is desired.
To validate the proposed method, multiple simulation results are provided.
\end{abstract}

\section{Introduction}
A multi-robot exploration problem is investigated in this paper. The multi-robot exploration in this context implies that a team of robots works together either to search for targets in the domain or to survey areas of interest. Therefore, a multi-robot exploration problem refers to how a team of autonomous robots covers the given domain. Illustrative examples for robot exploration problems are search and rescue, surveillance, reconnaissance, and site monitoring. 
It is worth noting that robot exploration problems are different from conventional robot path planning problems \cite{ryan2008exploiting,kala2012multi,ulusoy2013optimality} or 
robot coordination problems \cite{burgard2005coordinated,sheng2006distributed,wurm2008coordinated}.
The path planning problem is concerned about how to generate a robot trajectory such that a robot successfully arrives at a given goal point starting from an initial point.
In robot exploration problems, however, a goal point does not exist. Moreover, robot coordination problems are concerned about exploring an unknown environment for the mapping purpose, which is different from multi-robot explorations. In the following, in-depth literature reviews are conducted to provide insights on what studies have been done related to multi-robot explorations.

\subsection{Literature Review:}
Multi-robot exploration problems have been widely investigated over many decades. From a broad point of view, previous works related to multi-robot explorations can be classified into three different areas: multi-robot coverage path planning, multi-robot exploration and search, and ergodic exploration, depending on the method used in their work. 
In the following, their results and contributions to date are briefly reviewed.
\subsubsection{Multi-robot coverage path planning}
Coverage Path Planning (CPP) refers to a method to synthesize a robot path for passing over all points of an area or volume of interest.
Many different approaches have been developed to solve the multi-robot CPP problem, which is further divided into several sub-areas.

The cell decomposition method transforms the obstacle-free space into simple, non-overlapping regions called cells. The union of all the cells can be swept by a robot using simple motions. A lawnmower path is an example of such a simple motion through a zigzag pattern as in the literature \cite{azpurua2018multi}. 
The general cell decomposition technique is applied to the multi-robot case in \cite{xu2014efficient}, \cite{avellar2015multi}, where the generated graph is utilized to divide the sweeps among the team of robots by solving a vehicle routing problem. Recently, the authors in \cite{kapoutsis2017darp} proposed an algorithm that partitions the area of interest fairly among a team of robots considering their initial positions. 
A spanning tree is proposed in \cite{kim2014time} which focuses on time-synchronized coverage control of cooperative mobile robots. 

The geometric based approach is another branch of CPP based on visibility graphs that includes a set
of points and obstacles where the nodes represent the locations, and the edges are line segments that do not
pass through obstacles. This method has been used in many areas such as finding the shortest Euclidean path, and
polygonal area coverage. 
The broadly used geometric-based method in the multi-robot CPP is the Voronoi Diagrams \cite{yazici2013dynamic}, \cite{adaldo2017cooperative}.

Incremental random planners refer to sampling-based methods such as Rapidly exploring Random Trees (RRT) and Probabilistic  Road Map (PRM). These methods have been widely investigated with many variations including ant colony robot motion planning \cite{mohamad2005ant}, exploration of implicit roadmaps in multi-robot motion planning \cite{solovey2015finding}, parallel implementation of the RRT-based motion planner \cite{carpin2002parallel}, mutual information-based multiple autonomous underwater vehicles path planning \cite{cui2015mutual}, a multi-robot system for vacuum cleaning \cite{nikitenko2014multi}, a scalable and informed solution for asymptotically-optimal multi-robot motion planning \cite{shome2020drrt}, to list a few.

Other types of methods related to CPP problems are reward-based \cite{ranjbar2012multi, darrah2017optimized, li2018multi, palacios2017optimal} and Next Best View (NBV) approaches \cite{manjanna2018heterogeneous, miki2018multi, mirzaei2011cooperative, qin2019autonomous}. 
The reward-based approach is to reward agents for desired behavior in multi-robot CPP through a neural network, nature-inspired method, and hybrid algorithms. The advantages of this approach are nonlinear mapping, learning ability, and parallel processing. 
NBV approaches are generally used when no information
about the model exists. These approaches scale better
to complex real-world. 

Although there have been considerable research works related to the multi-robot CPP both in quality and quantity as described above, all of these methods have only focused on the entire coverage of the given domain while not taking into account relative importance or priority of areas of interest, making the multi-robot CPP far from the efficient exploration.

\subsubsection{Multi-robot exploration and search}
Hollinger \cite{hollinger2009efficient} proposed a multi-robot efficient search algorithm for a moving target in an indoor environment based on the Bayesian measurement update model. 
Under the assumption that an environment is known, the problem was formulated to choose multi-robot paths most likely to intersect with the path taken by the target.
The proposed method is then further extended to 
pursuit-evasion and autonomous search application \cite{chung2011search},
multi-robot coordination with periodic connectivity \cite{hollinger2012multirobot},
mapping, localization, and motion planning for multi-robot systems \cite{rone2013mapping},
sampling-based robotic information gathering algorithms \cite{hollinger2014sampling},
adaptive informative path planning \cite{lim2016adaptive}, and
online planning for multi-robot active perception \cite{best2018online}.
In some real applications (e.g., search and rescue), however, the environment may be only partially known or completely unknown, which limits the applicability of the proposed method. Moreover, robot dynamics and sensor models are not incorporated into the plan.  

Particle Swarm Optimization (PSO)-based approaches have been proposed for a multi-robot search algorithm \cite{pugh2007inspiring}. PSO has its basis in computer science to optimize a problem by making particles move around in the search-space according to simple mathematical formulae. PSO methods have been widely investigated and extended for multi-robot exploration and search \cite{li2011swarm, couceiro2011novel, tan2013research, zou2015particle, couceiro2014benchmark, senanayake2016search, kwa2020optimal}.
As is widely known, PSO is metaheuristic and hence, there is no theoretical guarantee for optimality.  The target search algorithm is greedy, implying that it relies only on the sensor measurements between multiple robots.

\subsubsection{Ergodic exploration}
Mathew and Mezi{\'c} \cite{GM-IM:11} addressed a multi-robot exploration problem based on the ergodicity that refers to system characteristics such that the time-averaged dynamics are equal to the given spatial average. In this work, a metric is defined to measure the ergodicity as the difference between the time-averaged multi-robot trajectory and the given spatial distribution. The Fourier basis functions are employed to facilitate the derivation of the ergodic control laws. This method has been further investigated and applied to many other in-depth research works such as 
coverage control of mobile sensors for a search of unknown targets \cite{surana2012coverage},
optimal planning for information acquisition \cite{silverman2013optimal},
trajectory optimization for continuous ergodic exploration \cite{miller2013trajectory},
real-time area coverage and target localization \cite{mavrommati2017real},
ergodic coverage in constrained environments \cite{ayvali2017ergodic},
decentralized ergodic control \cite{abraham2018decentralized},
receding-horizon multi-objective optimization for disaster response \cite{lee2018receding}, and 
ergodic flocking \cite{veitch2019ergodic}.

All of these works rely on the proposed result in \cite{GM-IM:11}, yet it contains the following issues.
The proposed result is developed for the centralized control scheme, which may not be desirable in reality. A computational issue arises in the implementation stage due to infinite numbers of the Fourier basis functions being used in the method \cite{GM-IM:11} and because of that, it is not clear how many Fourier basis functions are appropriate in implementation. 
A reference spatial Probability Density Function (PDF) is given as a static function and hence, it cannot cope with a time-varying scenario. 
Although some of these issues were tackled by several follow-up research works, no previous works have resolved all of them.
Finally, and most importantly, the ergodicity can be achieved only with infinite time, which is the fundamental limitation of the ergodic approach. This problem is fatal as robots have finite energy and hence, the ergodicity will never be attained in practice.

\subsection{Statement of Contributions:} 
Despite numerous research works into the development of multi-robot explorations, they have been limited and thus lacked a  significant contribution to the efficient exploration as explained above. 
In this paper, we propose a new approach for efficient multi-robot explorations based on the optimal transport theory. 
In \cite{kabir2020receding}, a preliminary result was introduced for an efficient single-robot exploration plan. This work has laid the foundation and opened up the possibility to generate an efficient robot trajectory based on the optimal transport theory. The optimal transport is also utilized to measure robot exploration efficiency. 
This preliminary work, however, was developed for a single robot and did not consider the majority of research works investigated in this paper such as multi-robot trajectory generation, non-overlapping issues between multiple robots, a decentralized control scheme, and a time-varying distribution scenario.
Therefore, it is considerably different from what we propose here.

The major contributions of this paper are as follows: 1) A new multi-robot exploration scheme is proposed for efficient explorations of a given domain based on the OT theory; 2) An energy level of a multi-robot system is taken into account in the plan. This is a critical issue as robots have finite energy in practice, which must be reflected in the plan to achieve the efficient exploration; 3) The proposed scheme is decoupled from robot dynamics, enlarging the applicability of the proposed method to various robot systems having heterogeneous platforms; 4) Both centralized and decentralized algorithms are developed to cope with more realistic scenarios, where the communication between agents are limited by some communication range constraints; 5) The proposed method is applicable to a time-varying distribution scenario, which is more desirable when the targets to be detected by a multi-robot system is supposed to move; 6) The performance of the proposed multi-robot exploration algorithm can be measured by the developed method, the upper bound of the optimal transport; and 7) The proposed method is computationally efficient compared to that from the similar research.

\subsection{Paper Outline and Notation:}
The remainder of this article is organized as follows.
Section~\ref{sec: prelim} introduces some preliminaries and problem description. The OT-based efficient exploration scheme for a single agent case is provided in Section \ref{sec: OT-method, single agent}. This result is extended to a multi-agent case, which is presented in Section \ref{sec: OT-method, multiple agent}. Section \ref{sec: time-varying} deals with a time-varying scenario for the given reference distribution. Simulation results are provided in Section \ref{sec: simulations} to support the validity of the proposed methods as well as to compare the performance.

\textit{Notation:} A Set of real and natural numbers, respectively, is denoted by $\mathbb{R}$ and $\mathbb{N}$. Further, $\mathbb{N}_0 = \mathbb{N}\cup\{0\}$. The symbols $\Vert \cdot \Vert$ and $^{T}$, respectively, denote the Euclidean norm and the transpose operator. 
The symbol $\#$ indicates the cardinality of a given set. The variable $t\in\mathbb{N}_0$ is used to denote a discrete time.

\section{Preliminary and Problem Description}\label{sec: prelim}

A given domain that needs to be covered by a multi-robot system may have differences in priority (or importance) and a team of heterogeneous robots is required to explore the domain according to the pre-specified priority of areas. This implies that agents need to spend more time in exploring some high-priority regions while covering low-priority regions with less time.

To achieve this goal, the optimal transport theory is employed as a tool. 
Traditionally, the optimal transport is to seek an optimal solution for a resource allocation problem \cite{villani2008optimal}. This optimal transport problem is formulated by the following Kantorovich's form:
\begin{itemize}
\item Kantorovich Optimal Transport problem:
\begin{align*}
\inf\left\{\int_{X\times Y}c(x,y)d\gamma(x,y)\vert \gamma\in\Gamma(\mu,\nu) \right\},
\end{align*}
\end{itemize}
where $\Gamma(\mu,\nu)$ denotes the collection of all probability measure with marginals $\mu$ on $X$ and $\nu$ on $Y$ and $c(x,y)$ is the distance between $x\in\mathbb{R}^{2}$ and $y\in\mathbb{R}^{2}$ (for two dimensional case).

Using the Euclidean distance $c(x,y) = \Vert x-y \Vert^p$ with $p^{th}$ order ($p\geq 1$), we introduce the Wasserstein distance \cite{villani2008optimal} of order $p$, which has been widely employed to broad dynamical systems including system analysis \cite{lee2014probabilistic}, \cite{lee2015performance}, \cite{lee2015analysis} as well as controller synthesis \cite{lee2014optimal}, \cite{lee2018optimal} problems.
\begin{itemize}
\item Wasserstein distance:
\begin{align*}
W_p(\mu,\nu) := \left(\inf_{\gamma\in\Gamma(\mu,\nu)}\int_{X\times Y} \Vert x - y\Vert^p d\gamma(x,y)\right)^{\frac{1}{p}},
\end{align*}
\end{itemize}
This Wasserstein distance describes the least amount of effort to transform one distribution $\mu$ into another one $\nu$.

The Hitchcock-Koopmans transportation problem is developed for the optimal transport problem in the discrete marginal case \cite{evans1997partial}, where $\mu$ and $\nu$ are represented by particles.
The following linear programming (LP) formulation of the transportation problem is equivalent to the Wasserstein distance in the sample point representation of given distributions.
\begin{itemize}
\item Linear Programming problem: (for $p=1$) 
\begin{equation}\label{eqn: LP}
  \begin{aligned}
    & \underset{\pi_{ij}}{\text{minimize}} & & \sum_{i,j}\pi_{ij}\Vert x_i-y_j \Vert\\
    & \text{subject to} & & \pi_{ij} \geq 0,\\
   	& & & \sum_{j=1}^{N}\pi_{ij} = m(x_i), i=1,2,\ldots,M,\\
	& & & \sum_{i=1}^{M}\pi_{ij} = n(y_j), j=1,2,\ldots,N,
\end{aligned}
\end{equation}
\end{itemize}
where $x_i$ and $y_j$ are the sample point of the ensemble, $m(x_i)$ and $n(y_j)$ are some non-negative constants representing the mass or weight of each particle in the ensemble. The variable $\pi_{ij}$ stands for the transport plan indicating how much weight has to be delivered from $x_i$ to $y_j$.
The optimal transport  plan $\pi_{ij}^*$ is to seek an optimal solution for the minimum effort to transport the weights. 

\begin{figure*}[t]
\begin{center}
\subfloat[given spatial distribution]{\includegraphics[scale=0.42]{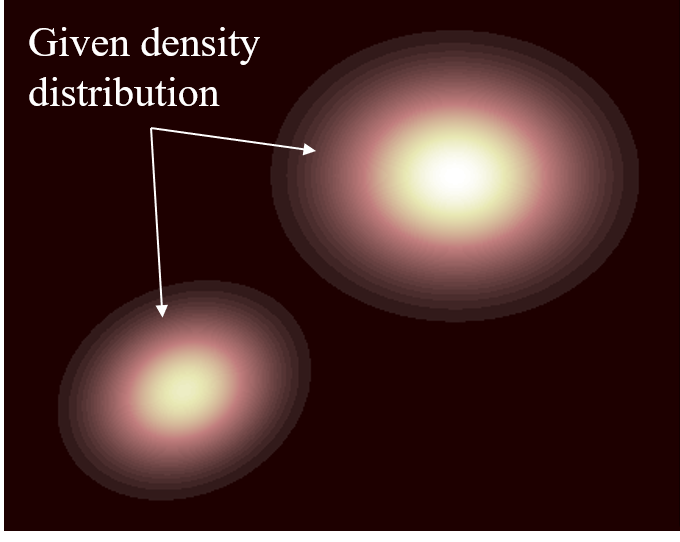}}\quad
\subfloat[sampling]{\includegraphics[scale=0.42]{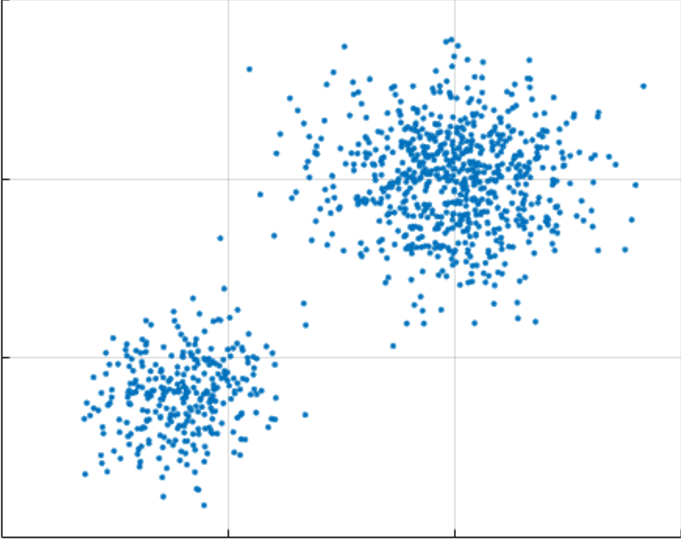}}\quad
\subfloat[efficient robot exploration ]{\includegraphics[scale=0.42]{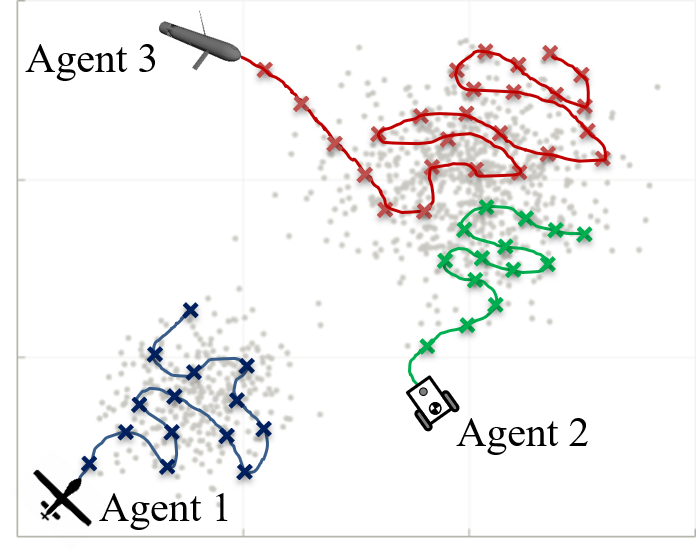}}
\end{center}
\caption{The procedure to generate the efficient multi-robot trajectory using the optimal transport theory}\label{fig: ergodic traj procedure}
\end{figure*}

In the multi-robot trajectory generation problem, the set of robot points $\{x_i\}$ need to be planned as these points are not determined yet. The Wasserstein distance in the LP form \eqref{eqn: LP} will be used as a tool to measure the difference between the two ensembles, one from the robot trajectories, $\{x_i\}$, and another from the given reference distribution, $\{y_i\}$.
Thus, the goal of this research is to generate the multi-robot trajectories such that the distribution formed by $\{x_i\}$ gets closer to the given reference distribution $\{y_i\}$, resulting in the efficient multi-robot exploration.
The schematic is presented in Fig. \ref{fig: ergodic traj procedure}. For the given spatial distribution (Fig. \ref{fig: ergodic traj procedure} (a)), the sampling process is carried out (Fig. \ref{fig: ergodic traj procedure} (b)), and then the multi-robot trajectories are generated to match $\{x_i\}$ and $\{y_j\}$ (Fig. \ref{fig: ergodic traj procedure} (c)).
From the mathematical perspective, it is described to generate robot trajectories $\{x_i\}$ such that $\displaystyle \sum_{i=1}\sum_{j=1}\pi_{ij}\Vert x_i - y_j\Vert \rightarrow 0$ with the given constraints in \eqref{eqn: LP}.

A simple and naive approach to achieve this goal is to make $\{x_i\}$ identical to $\{y_j\}$. However, this approach is infeasible due to the following reasons:
\begin{enumerate}
\item[1)] The robot's motion constraints may restrict the robot from visiting the sample point $y_j$.
\item[2)] Considering the energy limit for each robot, the total number of robot points denoted by $M$ is finite and thus, $M$ may be smaller than the total number of sample points given by $N$.
\item[3)] If $M\neq N$, it is not possible to match the robot points with the sample points.
\item[4)] For $M=N$, it may take an enormous amount of time for robots to explore the domain while following the generated trajectory connecting all the sample points if $N$ is very large. 
\end{enumerate}

In the following, the OT-based robot exploration scheme is provided to resolve all of the issues stated above.
\section{OT-based Efficient Exploration: A Single-Agent Case}\label{sec: OT-method, single agent}

This section provides a key idea for the efficient robot exploration based on the optimal transport theory. A single-agent scenario will be introduced first and then, it will be extended to the multi-robot scenario in the following section.

For an exploration mission with a given reference PDF, the domain needs to be covered by an agent that has finite energy. This energy constraint can be transformed into the total number of robot points, $M$, in discrete time. For a given number of robot points $M\in\mathbb{N}$, all the points are equally weighted by $m(x_i)$, where $m(x_i) = \frac{1}{M}$, $\forall i$. Similarly, for $N\in\mathbb{N}$ numbers of sample points representing the spatial distribution, each sample point $y_j$ has a uniform weight in the beginning, given by $n_0(y_j)=\frac{1}{N}$. Here, the weight $n_t(y_j)$ of a sample point $y_j$ decreases with time (discrete time $t\in\mathbb{N}_0$, particularly) as the robot explores the domain, making the weight $n_t(y_j)$ be a function of time $t$.

We consider that initially (at $t=0$), all the robot points $\{x_i\}_{i=1}^{M}$ are accumulated at the initial robot position $x_0$. As the robot changes its position from $x_0$ to $x_1$ in the next discrete-time step (the details and method will be introduced later), the weight assigned for the new position $x_1$ becomes $m(x_1) = \frac{1}{M}$. All the remaining weights $\frac{M-1}{M}$ for the undetermined robot points $\{x_i\}_{i=2}^{M}$ are carried by the robot and these undetermined points are considered to be located at the current robot position $x_1$. 
The following assumption is proposed to generalized the above description.

\begin{assumption}\label{assump: remaining weight}
Given the robot position $x_t$ at any discrete-time $t\in\mathbb{N}_0$, the weight $m(x_i)$, $i=1,2,\ldots, t$, for the past robot points is evenly given by $\frac{1}{M}$. The undetermined future robot points $\{x_i\}_{i=t+1}^{M}$ are all accumulated in the current robot position, $x_{t}$, which has remaining weights $\sum_{i=t+1}^{M}\left(\frac{1}{M}\right) = \frac{M-t}{M}$.
\end{assumption}

Based on Assumption \ref{assump: remaining weight}, the proposed efficient exploration scheme is introduced in the sequel.
Fig. \ref{fig: layer} depicts the multi-layer schematic of the proposed method to realize the efficient exploration.
\begin{figure}[h!]
\begin{center}
\includegraphics[scale=0.43]{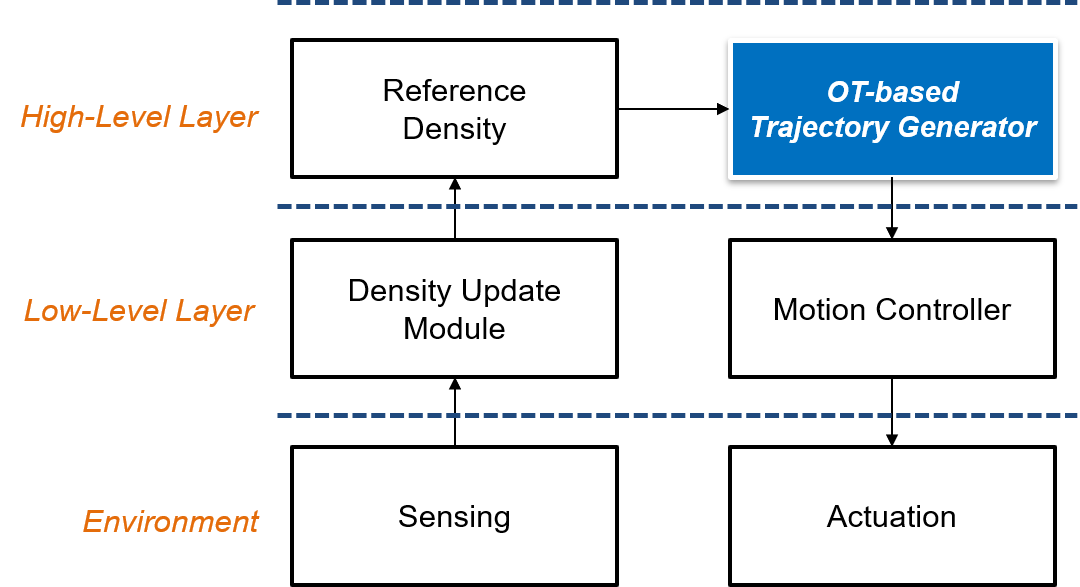}
\caption{Schematic of multiple layers with OT-based trajectory generator placed on high-level layer}\label{fig: layer}
\end{center}
\end{figure}
The major work is to develop the OT-based trajectory generator in the high-level layer. This trajectory generator receives the given reference PDF information and then, generates a trajectory for the robot to follow. The motion controller in the low-level layer is decoupled from the trajectory generator, which broadens the applicability of the proposed method to various robots having heterogeneous platforms. This implies that the proposed approach is platform-free and hence, the efficient exploration can be achieved in collaboration between various robot platforms such as unmanned aerial vehicles, ground robots, and unmanned underwater vehicles. As the robot explores the given domain, it collects data using onboard sensors from an environment. The density update module in the low-level layer performs this task. Finally, the reference PDF in the high-level layer will be updated through the density update module.

The optimal transport problem focuses on determining the non-negative optimal transport plan $\pi_{ij}^{\star}$ for given Euclidean distances $\Vert x_i - y_j\Vert$. Unlike conventional optimal transport problems in the LP form \eqref{eqn: LP}, the efficient robot exploration problem includes two parameters, $\pi_{ij}$ and $x_i$, both as the decision variables. 
This renders the efficient robot exploration problem much more difficult than the LP problem. In the following, a two-stage approach for the OT-based trajectory generator is proposed to solve this problem.


\subsection{Methodology: A Two-Stage Approach}

The proposed method consists of two stages: the next goal point determination and the weight update in a receding-horizon fashion. 
In the first stage, the robot determines where to go by considering a fixed number of sample points within a certain range. Then, the robot moves toward the next goal point. Once the robot arrives at a new position (which may differ from the goal position due to robot constraints), the robot updates the weight for each sample point in the weight update stage. This process is repeated until the remaining weights for all sample points become zero.


\subsubsection{Next goal point ($^{g}x_{t+1}$) determination stage}

At any given discrete-time step $t$, if the robot is located at $x_t$, the next robot goal position ${^{g}x}_{t+1}$ can be determined by the following steps. The robot selects $h$ numbers of sample points $y_j$ by creating a circle centered at the current robot position $x_t$ with an initial radius of $r_0$. The radius is incrementally increasing by $\delta$ until the robot finds $h$ numbers of sample points within the circle. Once the robot has found these points, a trajectory, connecting all the sample points in the circle starting from $x_t$, is generated as depicted in Fig. \ref{fig: next x} (a).

\begin{figure}[t!]
\begin{center}
\subfloat[]{\includegraphics[scale=0.29]{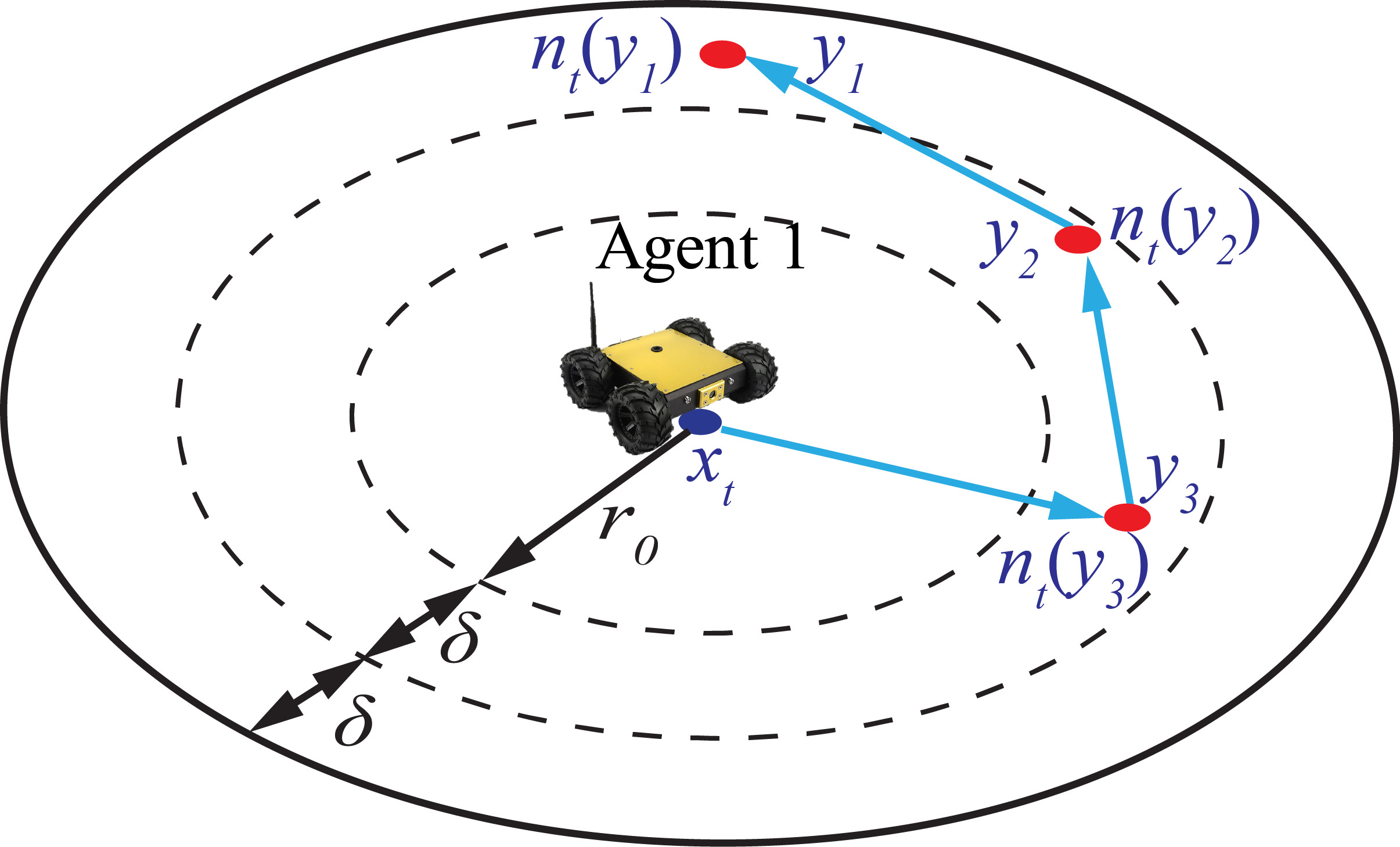}}\qquad\qquad
\subfloat[]{\includegraphics[scale=0.4]{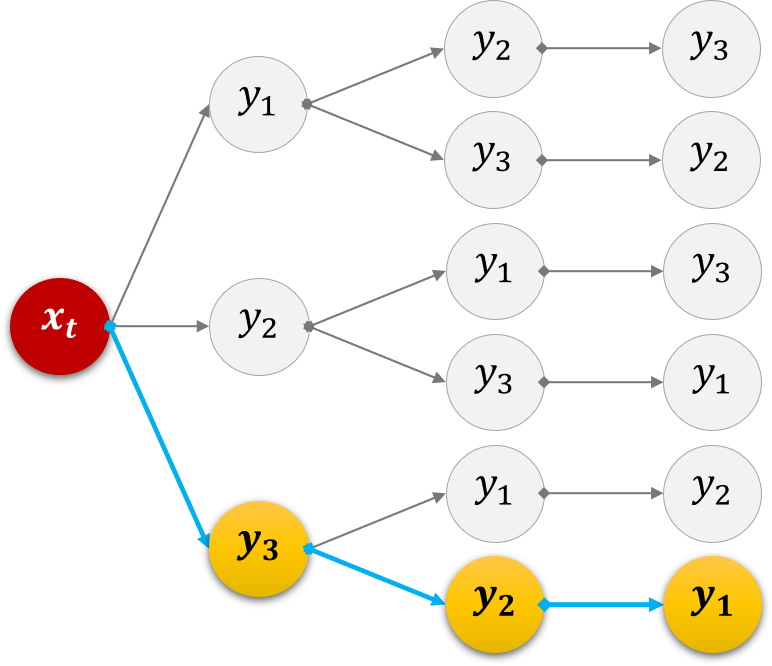}}
\end{center}
\caption{Schematic of the next goal point ${^{g}x}_{t+1}$ determination process: (a) increase the radius of the circle until $h$ numbers of sample points $y_j$ are found; (b) construct a tree associated with the found points $y_j$ and then select a particular path (blue arrows) that has a minimum cost}\label{fig: next x}
\end{figure}

A tree structure is constructed to connect all sample points in the circle starting from $x_t$.
In this case, there exist a total of ${h}!$ trajectories as exemplified in Fig. \ref{fig: next x} (b) ($h=3$ in this case). To determine the order of sample points for the robot to visit, a cost function is defined by
\begin{align}
    C(i) &= \dfrac{\lVert y_{\sigma_{t+1}}-x_{t}\rVert}{n_{t}(y_{\sigma_{t+1}})} + 
    \sum_{j=1}^{h-1}\dfrac{\lVert {y_{\sigma_{t+j+1}} - y_{\sigma_{t+j}}}\rVert}{n_{t}({y_{\sigma_{t+j+1}}})},\label{eqn: cost function}\\ i&=1,2,\ldots,{h} !,\nonumber
\end{align}
where $y_{\sigma_{t+j}}$, $j=1,2,\ldots,{h}$, are the sample points located within the circle with ${\sigma_{t+j-1}\neq\sigma_{t+j}}$, $\forall t\in\mathbb{N}_0$. 

The reason behind the cost function $C(i)$ defined by \eqref{eqn: cost function} is that we desire the robot to follow a trajectory that is short in terms of its total Euclidean distance as well as connects $y_j$ in the circle with a high weight value $n_t(y_j)$ first.

Given the definition $x_{t+1:t+h}:=\{x_{t+1},x_{t+2},\ldots,x_{t+h}\}$, 
the candidate trajectory for the robot $^{c}x_{t+1:t+h}(i)$, $i=1,2,\ldots,h!$, is obtained from the tree construction. Then, the $h$-step optimal trajectory $^{g}x_{t+1:t+h}$ is calculated by
\begin{align}
    ^{g}x_{t+1:t+h} = \{^{c}x_{t+1:t+h}({i^{\star})}\,\vert\, i^{\star}=\text{argmin}_{i} C(i)\}\label{eqn: x_{t+1}}
\end{align}

The robot decides the first point of $^{g}x_{t+1:t+h}$ as the next goal point, $^{g}x_{t+1}$, to visit. Again, the robot may or may not arrive at $^{g}x_{t+1}$ due to the robot's motion constraints.


\subsubsection{Weight update stage}
As soon as the robot arrives at a new position $x_{t+1}$, which might be different from $^{g}x_{t+1}$, the weight $n_{t+1}(y_j)$ associated with a sample point $y_j$ is updated using the following weight update law:
\begin{align}
n_{t+1}(y_j) = n_{t}(y_j)  - \pi_{(t+1)j}^{\star},\, \forall j\label{eqn: weight update}
\end{align}
where, $\pi_{(t+1)j}^{\star}$ is the optimal transport plan, depicting the weight distribution strategy from $x_{t+1}$ to each $y_j$. This plan is calculated from the solution of the LP problem stated below:

\begin{equation}\label{eqn: LP for weight update}
  \begin{aligned}
    & \underset{\pi_{(t+1)j}}{\text{minimize}} & & \sum_{j}\pi_{(t+1)j}\Vert x_{t+1} - y_j \Vert\\
    & \text{subject to} & & \pi_{(t+1)j} \geq 0,\,\,
	\sum_{j=1}^{N}\pi_{(t+1)j} = \dfrac{1}{M},\\
	& & & \pi_{(t+1)j} \leq \min\left( n_{t}(y_j), \frac{1}{M}\right), \, \forall j.
\end{aligned}
\end{equation}

The optimal solution for \eqref{eqn: LP for weight update} quantifies how much of the weight $\frac{1}{M}$ for the robot position $x_{t+1}$ is distributed to each $n_{t}(y_j)$ for the sample point $y_j$. The first constraint in \eqref{eqn: LP for weight update} guarantees that the transport $\pi_{(t+1)j}$ is non-negative. The second constraint is for the law of mass conservation that the total weight distributed from $x_{t+1}$ to $y_j$ must be equal to $\frac{1}{M}$. The last constraint is to make sure that the transportation $\pi_{(t+1)j}$ cannot exceed the given capacity for each point. This constraint is enforced by having the smaller value between $n_{t}(y_j)$ and $\frac{1}{M}$. 
Once the optimal transport plan $\pi_{(t+1)j}^{\star}$ is determined, the weight for each sample point is updated by \eqref{eqn: weight update}.

As $x_{t+1}$ is a single point, the analytic solution for \eqref{eqn: LP for weight update} is provided in the following proposition.
\begin{proposition}\label{prop: analytic solution}
The optimal solution for the LP problem \eqref{eqn: LP for weight update} is obtained by repeating
\begin{align*}
&\pi_{(t+1)j^{\star}} = 
\min\left(n_{t}(y_{j^{\star}}), m(x_{t+1})\right),\\
\text{ where } 
&j^{\star} = \argmin_{j\in\{j\vert n_{t}(y_j) > 0\}} \lVert x_{t+1} - y_j\rVert \\
&m(x_{t+1}) = m(x_{t+1}) - \pi_{(t+1)j^{\star}}\\
&n_{t}(y_{j^{\star}}) = n_{t}(y_{j^{\star}}) - \pi_{(t+1)j^{\star}}
\end{align*}
until $m(x_{t+1})$ becomes zero.
\end{proposition}

\begin{proof}
Given a single point $x_{t+1}$ in LP \eqref{eqn: LP for weight update}, the optimal transport plan for the robot is to deliver the maximum permissible weight to the closest points with positive weights in order as long as the weight $m(x_{t+1})$ remains positive.  
\end{proof}

This two-stage approach is repeated in a receding-horizon fashion, meaning that in each time step, the robot only considers $h$ numbers of sample points within the circle during the next goal point determination stage, followed by the weight update stage. Thus, the variable $h$ is given as the horizon length. If there are less than $h$ numbers of sample points having a positive weight, then only these points are considered.
As the robot explores the given domain, the weight of the sample points decreases and this process continues until all the weights of the sample points are completely depleted.

\subsection{Algorithm}
In Algorithm \ref{algorithm:1}, we provide the formal algorithm for the proposed efficient exploration planner.
Initially, the robot position $x_0$ and the reference PDF in the ensemble representation $\{y_j\}_{j=1}^{N}$ are given. Other parameters for the initialization are the number of robot points $M$, the horizon length $h$, and the radius increment $\delta$.
The radius of the circle $r$ increases until there exists a total of $h$ numbers of sample points in the set $\mathcal{R}(x_t,r)$ defined by all sample points $\{y_j\}$ having a positive weight $n_t(y_j) > 0$ within the circle centered at $x_t$ with a radius of $r$. This is given as the condition in line 4 of Algorithm \ref{algorithm:1}. Once satisfied, all possible trajectories are generated starting from $x_t$ to connect $h$ numbers of $y_j$ in the circle, followed by the cost function calculation \eqref{eqn: cost function}. Then, the next goal point $^{g}x_{t+1}$ is obtained by \eqref{eqn: x_{t+1}}, and the robot moves to the next point $x_{t+1}$ based on the given motion controller. In the last stage, the weight is updated by \eqref{eqn: weight update}. This process is repeated in a receding-horizon manner until the discrete-time becomes $M$.

\begin{algorithm}[h!]
\caption{Single-agent Intelligent Exploration Algorithm}\label{algorithm:1}
\begin{algorithmic}[1]
\State initialize $x_0$, $y_j$, $M$, $r_0$, $\delta$, ${h}$, $t\gets 0$
\While{$t\leq M$} 
\State initialize circle's radius by $r\gets r_0$
\While {$\#\mathcal{R}(x_t,r) \leq {h}$ and $n_{t}(y_j)>0$ }
\State $r\gets r + \delta$
\EndWhile
\State calculate the cost function $C(i)$ associated with all possible candidate trajectories $^{c}x_{t+1:t+h}(i)$
\State obtain $^{g}x_{t+1}$ from \eqref{eqn: x_{t+1}}
\State update the robot position $x_{t}$ with the given robot motion controller and the goal position $^{g}x_{t+1}$
\State update weights $n_t(y_j)$ by \eqref{eqn: weight update}
\State $t\gets t+1$
\EndWhile
\end{algorithmic}
\end{algorithm}


\subsection{Performance Measure using Wasserstein Distance }
For the performance measure that is the difference between the spatial PDF and another PDF formed by the robot trajectory, the Wasserstein distance in the LP form \eqref{eqn: LP} can be used to serve as a metric to quantify the performance. This metric provides information on how close the PDF from the robot trajectory is to the given reference PDF. For a large $M$, solving this LP problem in \eqref{eqn: LP} becomes computationally intractable due to the curse of dimensionality (i.e., total $M\times N$ numbers of computation is required). To circumvent this issue, the following theorem is proposed for the upper bound of the optimal solution in a computationally efficient way.

\begin{theorem}\label{thm: W_UB}
Consider the optimization problem \eqref{eqn: LP} under Assumption \ref{assump: remaining weight} with robot points $\{x_i\}_{i=1}^{t}$ determined by the proposed efficient exploration  algorithm. Then, at any time $t\in\mathbb{N}_0$, the Wasserstein distance $W(t)$ is upper bounded by
\begin{align}
    W(t) \leq \sum_{i=1}^{t}\tilde{W}(i) + \sum_{j=1}^{N}n_{t}(y_j) \cdot \lVert x_{t} - y_j\rVert,\label{eqn: W_UB}
\end{align}
where $n_{t}(y_j)$ is the current weight for each $y_j$ after the weight update law \eqref{eqn: weight update} and $\tilde{W}(i) := \minimize_{\pi_{ij}}\sum_{j=1}^{N}\pi_{ij}\Vert x_i - y_j\Vert$ subject to the same constraints in \eqref{eqn: LP for weight update}.
\end{theorem}

\begin{proof}
At any time $t\in\mathbb{N}_0$, the current and previous robot points $\{x_i\}_{i=1}^{t}$ as well as the remaining weights $n_t(y_j)$, $j=1,2,\ldots, N$, are given by the proposed algorithm. 
Under Assumption \ref{assump: remaining weight},  the future robot points are all accumulated at $x_{t}$.
Then, the Wasserstein distance at any time $t$ (constraints are omitted here) is upper bounded by
\begin{align*}
    W(t) &= \minimize_{\pi_{ij}}\sum_{i=1}^{M}\sum_{j=1}^{N}\pi_{ij}\Vert x_{i} - y_j \Vert\\
    &\leq \minimize_{\pi_{ij}}\sum_{i=1}^{t}\sum_{j=1}^{N}\pi_{ij}\Vert x_i - y_j\Vert + \qquad\qquad\qquad\qquad\qquad\\
    & \qquad\qquad\qquad \minimize_{\pi_{ij}}\sum_{i=t+1}^{M}\sum_{j=1}^{N}\pi_{ij} \lVert x_{i} - y_j\rVert\\
    &\leq \sum_{i=1}^{t}\underbrace{\left(\minimize_{\pi_{ij}}\sum_{j=1}^{N}\pi_{ij}\Vert x_i - y_j\Vert\right)}_{=\tilde{W}(i)} + \qquad\qquad\qquad\qquad\qquad\\
    & \qquad\qquad\qquad\qquad\qquad \sum_{j=1}^{N}n_{t}(y_j) \cdot \lVert x_{t} - y_j\rVert,
\end{align*}
where the last inequality holds by Assumption \ref{assump: remaining weight} and the mass conservation law.
\end{proof}

The upper bound of the Wasserstein distance can be computed at any time $t\in\mathbb{N}_0$ from \eqref{eqn: W_UB}, which only requires $\tilde{W}(t)$ computation, followed by the weight update \eqref{eqn: weight update} for the computation of the second term in \eqref{eqn: W_UB}.
The value for $\tilde{W}(t)$ is analytically obtained by Proposition \ref{prop: analytic solution} and the upper bound is calculated recursively as the values for $\tilde{W}(i)$, $i=1,2,\ldots,t-1$, are already computed, and thus known from the previous time step.


\section{Multi-Agent Exploration}\label{sec: OT-method, multiple agent}
To maximize exploration efficiency, it is better to utilize a team of agents instead of a single agent because deploying multiple agents will reduce the time to explore the given domain, and hence will improve exploration performance.
The previous OT-based algorithm provides insight into how a single-agent system performs the exploration of the given domain, which can be further extended to multi-agent explorations. Depending on communications between agents and control method, two different scenarios will be considered: a centralized and decentralized case. The centralized scenario is developed under the assumption that there exists a supervisory agent that can receive all information for the weight from all agents, updates the weight for common weight $n_t(y_j)$ values, and share it with all agents. This is impractical as communications between agents may not be available in some cases. Moreover, the synchronization issue arises in reality. Thus, the decentralized multi-agent exploration scheme is provided as well to cope with this concern. In the following, both scenarios will be introduced in details. 


\subsection{Centralized Multi-Agent Case}
For the centralized case, a supervisory agent that communicates with all other agents collects information about the weight from all of them, update a common weight, and then transmit this information to all agents. The total number of agents is represented by $n_a\in\mathbb{N}$ and the sample point representation of the reference PDF, $\{y_j\}_{j=1}^{N}$, is assumed to be identical across all agents initially.

In the centralized multi-agent case, the total number of robot points $M$ is given such that
\begin{align}
M = n_a t_e\label{eqn: effective timestep}
\end{align}
where $t_e\in\mathbb{N}$ denotes the effective number of time steps of each agent for exploration. The variable $t_e$ can be interpreted as the maximum allowable time steps for the exploration based on the energy level of the agents. It is assumed that all agents have the same energy level initially.
For the given initial robot positions $\{x_0^k\}_{k=1}^{n_a}$, the number of multi-robot points at any time step $t$ is calculated by $n_at$, where $t\leq t_e$. 

The formal algorithm for the centralized multi-agent scenario is provided in Algorithm \ref{algorithm:2}. 
\begin{algorithm}[h!]
\caption{Centralized Multi-Agent Exploration Algorithm}\label{algorithm:2}
\begin{algorithmic}[1]
\State initialize $x_0^k$, $y_j$, $M$, $N$, $r_0$, $\delta$, ${h}$, $n_a$, $t\gets 0$
\While{$t\leq t_e$}
\State \textbf{each agent implements the following}
\For{$k \gets 1$ to $n_a$} 
\State initialize circle's radius by $r\gets r_0$
\While {$\#\mathcal{R}(x_t^k,r) \leq {h}$ and $n_{t}^k(y_j)>0$}
\State $r\gets r + \delta$
\EndWhile
\State calculate the cost function $C^k(i)$ associated with all possible candidate trajectories $^{c}x^k_{t+1:t+h}(i)$
\State obtain $^{g}x^k_{t+1}$ from \eqref{eqn: x_{t+1}}
\State update the robot position $x^k_{t}$ with the given robot motion controller and the goal position $^{g}x^k_{t+1}$
\State update the individual weight $n_t^k(y_j)$ by \eqref{eqn: weight update}
\EndFor
\State \textbf{supervisory agent does the following}
\State \qquad receives information about $n_t^k(y_j)$ from all agents
\State \qquad updates the common weight $n_t(y_j)$ from \eqref{eqn: general weight update}
\State \qquad transmits $n_t(y_j)$ to all corresponding agents
\State \textbf{each agent} receives $n_t(y_j)$ from supervisory agent and $n_t^k(y_j) \gets n_t(y_j)$
\State $t\gets t+1$
\EndWhile
\end{algorithmic}
\end{algorithm}
Similar to the single-agent algorithm, each agent selects $h$ numbers of sample points within the circle that are centered at each agent's current position $x_t^k$ by following steps from 4 to 8 in Algorithm \ref{algorithm:2}. 
Each agent generates all possible trajectories and calculates the cost $C^k(i)$ from \eqref{eqn: cost function}. The next goal points $^{g}x^k_{t+1}$ for the agents are determined from \eqref{eqn: x_{t+1}} and the robots change their position to the next point $x^k_{t+1}$ driven by the given motion controller. At the new position, each agent updates the weight $n_t^k(y_j)$ using \eqref{eqn: weight update}.  Once the individual weight is updated by each agent, the supervisory agent receives the updated information from all of them. Then, the common weight is updated using the following equation:
\begin{align}\label{eqn: general weight update}
    & n_t(y_j) = \min(n_t^k(y_j)), \qquad k=1,2,\ldots,n_a
\end{align}
This updated general weight information will be shared with all agents and updated by each agent as follows:
\begin{align}\label{eqn: general weight share}
    & n_t^k(y_j) = n_t(y_j), \qquad k=1,2,\ldots,n_a
\end{align}

In this way, each agent exactly knows how other areas are covered by other agents. Therefore, the sample point whose weight is already depleted by other agents will not be revisited because of the information sharing between the agents.
These series of actions as presented in Algorithm \ref{algorithm:2} are performed until the current time step becomes $t_e$.

\subsection{Decentralized Multi-Agent Case}
The underlying assumption for the centralized multi-agent exploration is that the supervisory agent can communicate with all the other agents, regardless of the distances between them. However, in real applications, the communication range determines the agents' ability to communicate with the supervisory agent. Shorter communication range will interrupt this communication and thus the information sharing between the agents and the supervisor, resulting in failure of the centralized exploration scheme. Further, it is known that a centralized control scheme is more vulnerable to a single point of failure (i.e., a breakdown of the supervisory agent will lead to the failure of the whole system).
Considering these issues, the decentralized multi-agent exploration scheme is provided as an alternative but a more realistic solution.

Due to the limited communication range, the agents can only exchange information if any two agents are within the given communication range. Otherwise, each agent explores the given domain independently exactly like a single agent and shares information with other agents when they are within the communication range. 
The formal algorithm for the decentralized exploration strategy is provided in Algorithm \ref{algorithm:3}. 
\begin{algorithm}[h!]
\caption{Decentralized Multi-Agent Exploration Algorithm}\label{algorithm:3}
\begin{algorithmic}[1]
\State initialize $x_0^k$, $y_j$, $M$, $N$, $r_0$, $r_{\text{comm.}}$, $\delta$, ${h}$, $n_a$, $t\gets 0$
\While{$n_t^k(y_j) >  0,\,  \forall j, \forall k $} 
\For{$k \gets 1$ to $n_a$} 
\If{$d_{kq} \leq r_{\text{comm.}}$}
\State update weight information from \eqref{eqn: decentralized information exchange}
\EndIf
\State initialize circle's radius by $r\gets r_0$
\While {$\#\mathcal{R}(x_t^k,r) \leq {h}$ and $n_{t}^k(y_j)>0$}
\State $r\gets r + \delta$
\EndWhile
\State calculate the cost function $C^n(i)$ associated with all possible candidate trajectories $^{c}x^k_{t+1:t+h}(i)$
\State obtain $^{g}x^k_{t+1}$ from \eqref{eqn: x_{t+1}}
\State update the robot position $x^k_{t}$ with the given robot motion controller and the goal position $^{g}x^k_{t+1}$
\State update weights $n_t^k(y_j)$ by \eqref{eqn: weight update}
\EndFor
\State $t\gets t+1$
\EndWhile
\end{algorithmic}
\end{algorithm}
Initially, the positions of the agents $\{x_0^k\}_{k=1}^{n_a}$, the sample point distribution $\{y_j\}$, communication range $r_{comm.}$, number of robot points $M$, the horizon length $h$, and the radius increment $\delta$ are known. All processes are exactly the same as the single-agent case, except when one agent encounters another agent(s) within the given communication range.
If the agent $k$ finds the agent $q$ within the range $r_{\text{comm.}}$ at a given time $t$ (i.e., distance $d_{kq} \leq r_{\text{comm.}} $), then the weight information for the sample points $y_j$, $j=1,2,\ldots, N$ is exchanged and updated using the following rule: 
\begin{align}\label{eqn: decentralized information exchange}
   &n^k_t(y_j) =  n^q_t(y_j) = \min(n^k_t(y_j), n^q_t(y_j)),\\ 
   &k, q \in \{ 1,2,\ldots,n_a\} \text{ for } k\neq q\nonumber
\end{align}
By exchanging information, each agent grasps what sample points are already covered by other agents, leading to an efficient exploration while avoiding area overlaps between agents within the communication range.

It should be pointed out that in the single-agent and centralized multi-agent cases, the number of total time steps is straightforward which is not the case for the decentralized scheme. The duration of the decentralized exploration will depend on the communication range and how frequently the agents communicates with each other. If the communication range covers the entire domain, the agents are capable of communicating with each other at every time steps, technically rendering it a centralized exploration. In this case, the duration of the exploration is $t_e$, where $t_e = \frac{M}{n_a}$. On the contrary, if no agent can communicate with others during the exploration, all agents will act exactly like a single agent and explore the domain completely independently. Here, the total exploration time is the same as the number of total robot points $M$. Clearly, $t_e$ and $M$ are lower and upper limits of the actual decentralized exploration time, respectively. The exploration will continue until the weight $n_t^k(y_j)$ becomes zero for all sample points.

\subsection{Performance Measure}
To measure the performance of both the centralized and decentralized schemes, the upper bound of the Wasserstein distance is provided. 
For large $M$ and $N$, the actual Wasserstein distance computation becomes intractable, however, the following upper bound computation will provide a computationally efficient way to compute the upper bound of the Wasserstein distance, which can be used to measure how a team of robots performs a given exploration mission efficiently.
As the centralized and decentralized schemes work differently, the upper bound computation is developed for each scheme separately.

\subsubsection{Centralized case}
In the centralized case, all agents share the weight information at any time step. Thus, any agent can compute the Wasserstein distance, which is calculated by $W(t) = \minimize_{\pi_{ij}}\sum_{i=1}^{M=n_at_e}\sum_{j=1}^{N}\pi_{ij}\Vert x_{i} - y_j \Vert$ under Assumption \ref{assump: remaining weight}, for the centralized multi-agent system.
To facilitate this computation, an upper bound is provided in the following theorem.

\begin{theorem}\label{thm: Central_W_UB}
Consider the optimization problem \eqref{eqn: LP} under Assumption \ref{assump: remaining weight} with robot points $\{x_i\}_{i=1}^{t}$ determined by the proposed efficient exploration  algorithm. Then, at any time $t\in\mathbb{N}_0$, the Wasserstein distance $W(t)$ is upper bounded by
\begin{align}
    W(t) \leq \sum_{k=1}^{n_a}\sum_{i=1}^{t}\tilde{W}^k(i) + \sum_{k=1}^{n_a}\sum_{j=1}^{N}n_{t}(y_j) \cdot \lVert x_{t}^{k} - y_j\rVert,\label{eqn: Central_W_UB}
\end{align}
where $n_{t}^{k}(y_j)$ is the current weight for each $y_j$ after the weight update law \eqref{eqn: weight update} and $\tilde{W}^{k}(i):=\minimize_{\pi_{ij}^k}\sum_{j=1}^{N}\pi_{ij}^k\Vert x_i^{k} - y_j\Vert$ subject to the same constraints in \eqref{eqn: LP for weight update}.
\end{theorem}

\begin{proof}
At any time $t\in\mathbb{N}_0$, the current and previous robot points $\{x_i\}_{i=1}^{t}$ as well as the remaining weights $n_t(y_j)$, $j=1,2,\ldots, N$, are given by the proposed algorithm. 
Under Assumption \ref{assump: remaining weight},  the future robot points are all accumulated at $x_{t}$.
Then, the Wasserstein distance at any time $t$ (constraints are omitted here) is upper bounded by
\begin{align*}
    W(t) &= \minimize_{\pi_{ij}}\sum_{i=1}^{M=n_at_e}\sum_{j=1}^{N}\pi_{ij}\Vert x_{i} - y_j \Vert\\
    &\leq \minimize_{\pi_{ij}^{k}}\sum_{k=1}^{n_a}\sum_{i=1}^{t}\sum_{j=1}^{N}\pi_{ij}^{k}\Vert x_i^{k} - y_j\Vert + \\
    &\qquad\qquad \minimize_{\pi_{ij}^{k}}\sum_{k=1}^{n_a}\sum_{i=t+1}^{t_e}\sum_{j=1}^{N}\pi_{ij}^{k} \lVert x_{i}^{k} - y_j\rVert\\
    &\leq \sum_{k=1}^{n_a}\sum_{i=1}^{t}\underbrace{\left(\minimize_{\pi_{ij}^k}\sum_{j=1}^{N}\pi_{ij}^k\Vert x_i - y_j\Vert\right)}_{=\tilde{W}^{k}(i)} + \\
    &\qquad\qquad
    \sum_{k=1}^{n_a}\sum_{j=1}^{N}n_{t}(y_j) \cdot \lVert x_{t}^{k} - y_j\rVert,
\end{align*}
where the last inequality holds by Assumption \ref{assump: remaining weight} and the mass conservation law.
\end{proof}

At any time $t\in\mathbb{N}_0$, the upper bound of the Wasserstein distance is obtained by \eqref{eqn: W_UB}, which only requires $\tilde{W}(t)$ computation (by Proposition \ref{prop: analytic solution}), followed by the weight update \eqref{eqn: weight update} because the values for $\tilde{W}(i)$, $i=1,2,\ldots,t-1$, are already computed and known from the previous time step.

\subsubsection{Decentralized case}
The upper bound of the Wasserstein distance for the decentralized case is computed similarly to the centralized case. The only difference is that all agents do not share the weight information and each agent needs to calculate its own upper bound for the performance measure. The agent $k$ receives the weight information from its neighboring agents within the communication range. The set of neighboring agents within the communication range for the agent $k$ is denoted by $\mathcal{N}_k$. The Wasserstein distance for the agent $k$ is then computed by $ W^k(t) = \minimize_{\pi_{ij}^{k}}\sum_{k\in\mathcal{N}_k}\sum_{i=1}^{t_e}\sum_{j=1}^{N}\pi_{ij}^{k}\Vert x_{i}^{k} - y_j \Vert$ and the upper bound for this value is provided below.

\begin{theorem}\label{thm: Central_W_UB}
Consider the optimization problem \eqref{eqn: LP} under Assumption \ref{assump: remaining weight} with robot points $\{x_i\}_{i=1}^{t}$ determined by the proposed efficient exploration algorithm. Then, at any time $t\in\mathbb{N}_0$, the Wasserstein distance $W^k(t)$ for the agent $k$ is upper bounded by
\begin{align}
    W^{k}(t) \leq \sum_{k\in\mathcal{N}_k}\sum_{i=1}^{t}\tilde{W}^k(i) + \sum_{k\in\mathcal{N}_k}\sum_{j=1}^{N}n_{t}^{k}(y_j) \cdot \lVert x_{t}^{k} - y_j\rVert,\label{eqn: Central_W_UB}
\end{align}
where $n_{t}^{k}(y_j)$ is the current weight for each $y_j$ after the weight update law \eqref{eqn: weight update} and $\tilde{W}^{k}(i):=\minimize_{\pi_{ij}^k}\sum_{j=1}^{N}\pi_{ij}^k\Vert x_i^{k} - y_j\Vert$ subject to the same constraints in \eqref{eqn: LP for weight update}.
\end{theorem}

\begin{proof}
At any time $t\in\mathbb{N}_0$, the current and previous robot points $\{x_i\}_{i=1}^{t}$ as well as the remaining weights $n_t(y_j)$, $j=1,2,\ldots, N$, are given by the proposed algorithm. 
Under Assumption \ref{assump: remaining weight},  the future robot points are all accumulated at $x_{t}$.
Then, the Wasserstein distance at any time $t$ (constraints are omitted here) is upper bounded by
\begin{align*}
    W^k(t) &= \minimize_{\pi_{ij}^{k}}\sum_{k\in\mathcal{N}_k}\sum_{i=1}^{t_e}\sum_{j=1}^{N}\pi_{ij}^{k}\Vert x_{i}^{k} - y_j \Vert\\
    &\leq \minimize_{\pi_{ij}^{k}}\sum_{k\in\mathcal{N}_k}\sum_{i=1}^{t}\sum_{j=1}^{N}\pi_{ij}^{k}\Vert x_i^{k} - y_j\Vert + \\
    &\qquad\qquad\minimize_{\pi_{ij}^{k}}\sum_{k\in\mathcal{N}_k}\sum_{i=t+1}^{t_e}\sum_{j=1}^{N}\pi_{ij}^{k} \lVert x_{i}^{k} - y_j\rVert\\
    &\leq \sum_{k\in\mathcal{N}_k}\sum_{i=1}^{t}\underbrace{\left(\minimize_{\pi_{ij}^k}\sum_{j=1}^{N}\pi_{ij}^k\Vert x_i - y_j\Vert\right)}_{=\tilde{W}^{k}(i)} + \\
    &\qquad\qquad\sum_{k\in\mathcal{N}_k}\sum_{j=1}^{N}n_{t}^{k}(y_j) \cdot \lVert x_{t}^{k} - y_j\rVert,
\end{align*}
where the last inequality holds by Assumption \ref{assump: remaining weight} and the mass conservation law.
\end{proof}

As described in the performance measure for the single-agent scenario, the upper bound of the Wasserstein distance can be calculated in a computationally efficient manner using the analytic solution in Proposition \ref{prop: analytic solution}. Therefore, the performance can be monitored for both centralized and decentralized cases in real time.
\section{Time-Varying Distribution: Random Walk Model}\label{sec: time-varying}

In previous sections, the reference PDF is assumed to be stationary, which may not be realistic in some cases such as search, surveillance, and monitoring missions where targets are moving in the domain. Thus, it is more natural to consider the time-varying reference PDF to cope with more complicated but realistic scenarios. One of the benefits of using the OT-based efficient exploration scheme is that the spatio-temporal evolution of the given reference PDF can be incorporated into the plan since the reference PDF is represented by the ensemble. For any given dynamics used for the time-varying PDF scenario, each sample point $y_j$ can be updated based on the given dynamics. While each sample point is evolving with the given dynamics in each discrete-time step, a team of agents will be able to reflect the spatio-temporal evolution of the reference PDF in the proposed scheme.

Although the positions of sample points can be updated according to any dynamics for given applications (e.g., wind/ocean waves and temperature propagation), the uncorrelated random walk is considered here to model the movement of targets and the sample points.
In general, a correlated random walk model \cite{bergman2000caribou, bovet1988spatial, kadota2011analysis, kareiva1983analyzing} has been used to describe the movement of foraging animals, fishes, and insects, however, an uncorrelated random walk dynamics is applied for simplicity.  
We can model the stochastic movement of the animal herd with $N_h\in\mathbb{N}$ numbers of animals as follows. At given time $t$ with the location of $k^{th}$ member of a foraging herd $z_{t}^{k}$, the updated random location at the next time step $t+1$ can be found by
\begin{align}\label{eqn: animal position update}
    z_{t+1}^{k} = z_{t}^k + vu
\end{align}
where $v$ is a diffusion rate constant that indicates the diffusion rate of the foraging herd and $u = [u_x, u_y]$ represents two uniformly distributed random numbers such that $-1\leq u_x, u_y \leq 1$.

Since the animal positions are changing with time, the probability distribution for their locations also needs to be updated. Instead of adjusting the distribution, we change the position of the sample points $\{y_j\}$ similar to the animal position update. For a $j^{th}$ sample point $y_{j}(t)$ at time $t$, the location at the next time step $t+1$ is updated by
\begin{align}\label{eqn: sample point update}
    y_{j,t+1} = y_{j,t} + vw
\end{align}
where $v$ is the diffusion rate constant as in \eqref{eqn: animal position update} and $w = [w_x, w_y]$ is a vector for two uniformly distributed random numbers with $-1\leq w_x, w_y \leq 1$. The sample point updates will be performed in each time step, followed by the multi-robot trajectory generation.

\section{Simulations}\label{sec: simulations}

Several simulations were conducted to verify the technical soundness of the proposed efficient multi-robot exploration schemes. 
Although the proposed methods apply to any robot dynamics (e.g., the result for the nonlinear unicycle dynamics is presented in the previous work \cite{kabir2020receding}) as the multi-robot trajectory generator is separated from it, the robot dynamics considered here is the first-order as follows:
\begin{align}\label{eqn: first order continuous}
    \dot{x}(t) = \frac{dx(t)}{dt} = u(t)
\end{align}
where $x(t)\in\mathbb{R}^2$ is the continuous planar position of the agent and $u(t)\in\mathbb{R}^2$ is the instantaneous velocity as the control input for the first-order dynamics. 
The main reason for the simple first-order dynamics being chosen is for the comparison between the proposed OT-based method and the SMC method as the SMC method is developed for first-order and second-order dynamics. 

The discrete-time counterpart of \eqref{eqn: first order continuous} with the proposed control input $u$ is 
\begin{align}
    x_{t+1} = x_t + u\Delta t = x_t + u_{max}\frac{^gx_{t+1} - x_t}{||^gx_{t+1} - x_t||}\Delta t\label{eqn: robot motion control}
\end{align}
where $x_{t} = [\mathsf{x}_t, \mathsf{y}_t]^{T}$ is the robot position with $\mathsf{x}_t,\mathsf{y}_t\in\mathbb{R}$, $u_{max}$ is the maximum speed attainable by the robot, $\Delta t$ is the time interval for the discretization, and $^gx_{t+1}$ is the goal point for the next time step determined by \eqref{eqn: x_{t+1}} in the next goal point determination stage. 

All simulations were carried out by the computer platform, a $64$-bit quad core Intel $8^{th}$ gen Core i$5-8250$U @ $1.60$ GHz processor and DDR-$4$ $16$GB RAM. MATLAB $2016$a was used to simulate the proposed methods as the software.

\subsection{Centralized Multi-Agent Case}

The reference PDF is given as a mixture of Gaussian in the following form: $\rho = \sum_{i=1}^{4}\alpha_i\rho_i$, where $\alpha_i = 0.25$, $\forall i$, and each component-wise bivariate Gaussian distribution, $\rho_i$, has the mean and covariance given by
\begin{align*}
    \mu_1 &=[300, 1200]^{T},  \mu_2=[1000, 900]^T, \\
    \mu_3 &= [700, 300]^T, 
    \mu_4 = [1500, 1000]^T\\
    \Sigma_1 &= \begin{bmatrix}
       8000 & 0 \\ 
          0 & 4800
    \end{bmatrix}, 
    \Sigma_2 = \begin{bmatrix}
        3200 & 0\\ 
          0 & 4800
    \end{bmatrix}, \\
    \Sigma_3 &= \begin{bmatrix}
        6000 & 0 \\ 
           0 & 4800
    \end{bmatrix}, 
    \Sigma_4 = \begin{bmatrix}
        1500 & 0\\ 
          0 & 5000
    \end{bmatrix}
\end{align*}

Also, the simulation parameters used in this simulation are
\begin{itemize}
  \item Domain size: $1800  \times 1600 $
  \item Number of agents: $n_a = 5$
  \item Maximum number of robot steps: $M=5000$
  \item Effective time steps: $t_e = 1000$
  \item Number of sample points for the mixture of Gaussian distribution: $N = 2000$
  \item Initial robot positions: 
\begin{table}[!h]
\centering
\begin{tabular}{|c|c|c|c|c|c|}
\hline
  & agent 1    & agent 2    & agent 3    & agent 4   & agent 5    \\ \hline
$\mathsf{x}$ & 1000 & 1600 & 1400 & 300 & 600  \\ \hline
$\mathsf{y}$ & 1200 & 800  & 1300 & 800 & 1200 \\ \hline
\end{tabular}
\end{table}
  \item Maximum velocity of the robots: $100$
\end{itemize}

The snapshots of simulation for the five-agent exploration is described in Fig. \ref{fig: Centralized}. 
From the reference PDF given above, a total of $300$ randomly distributed targets were generated with red plus marks in Fig. \ref{fig: Centralized} (a). 
The agents are assumed to equip onboard sensors with a sensing range limit, $r_{\text{sensing}} = 15$, to detect targets. Although the sample points and the targets are generated from the identical distribution, positions of each set are different from each other. Moreover, all agents share the sample point information $\{y_j\}$, whereas the target positions are completely unknown to all of them. The agents can detect a target(s) only if a target(s) is within the sensing range. In Fig. \ref{fig: Centralized} (a)-(e), the detected targets are represented by black dots. 

\begin{figure*}[t]
\centering
\subfloat[$t=0$]{\includegraphics[scale=0.29]{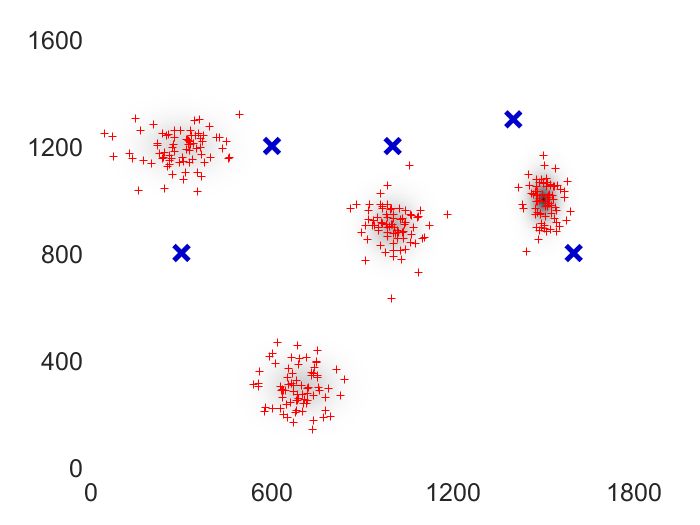}}
\subfloat[$t=300$]{\includegraphics[scale=0.29]{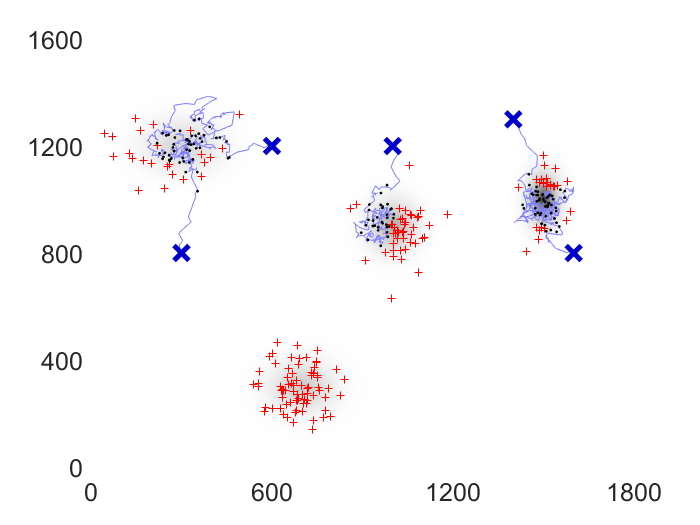}}
\subfloat[$t=600$]{\includegraphics[scale=0.29]{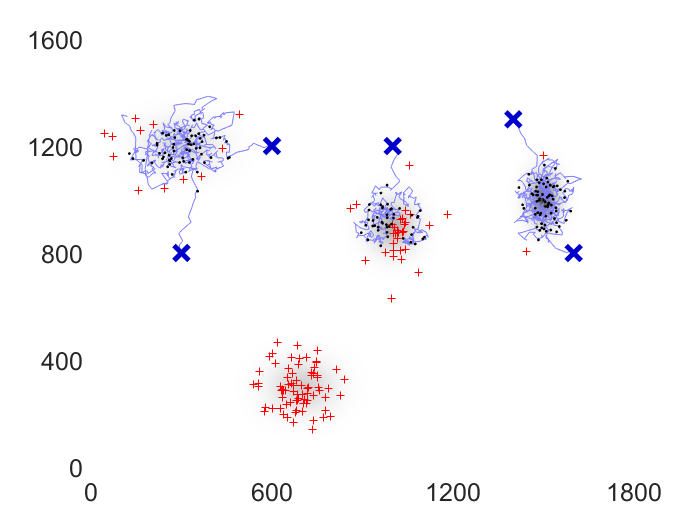}}\\
\subfloat[$t=800$]{\includegraphics[scale=0.29]{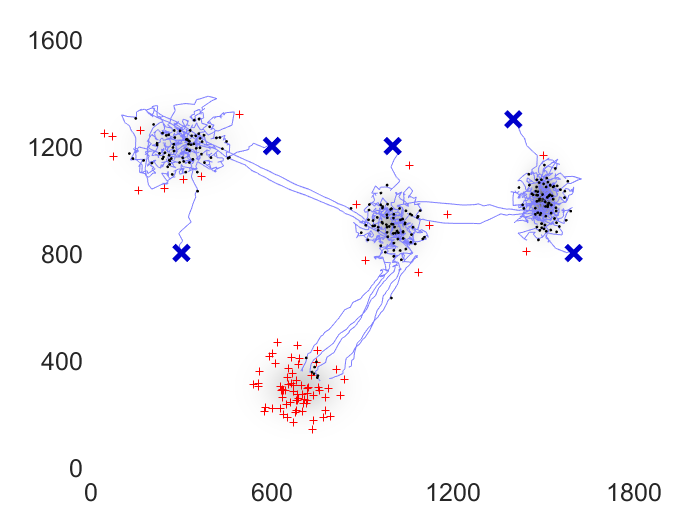}}
\subfloat[$t=1000$]{\includegraphics[scale=0.29]{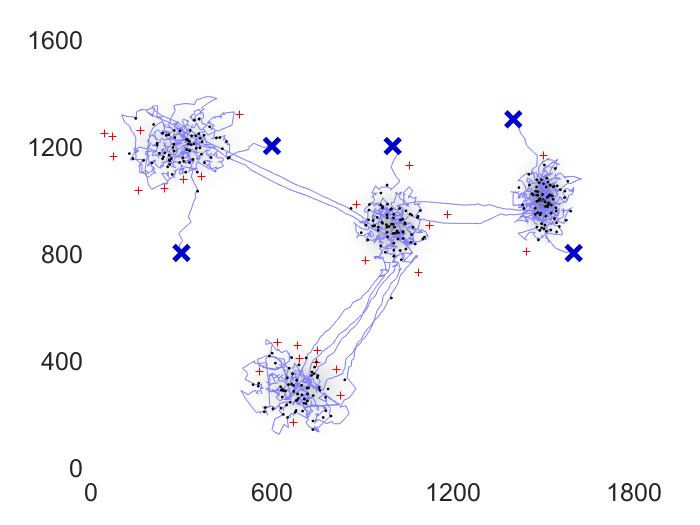}}
\subfloat[$W_{UB}$]{\includegraphics[scale=0.2]{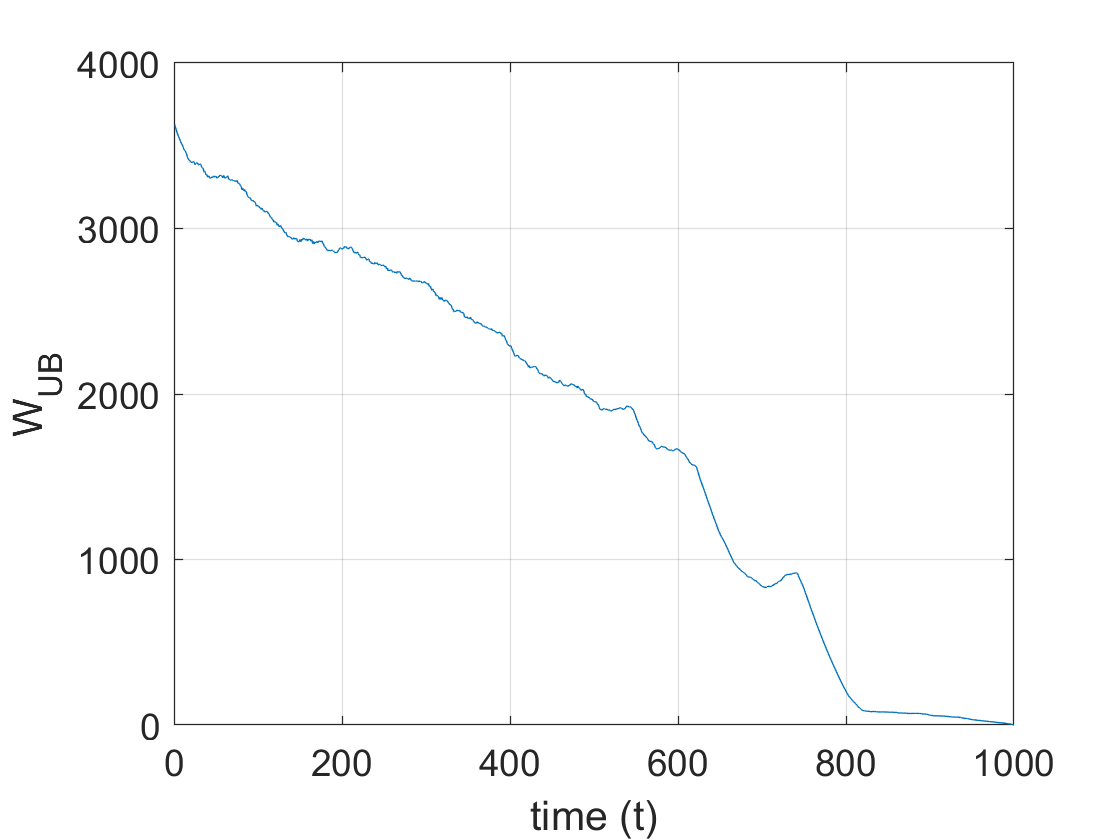}}
\caption{Snapshots of the centralized multi-agent exploration  (a)-(e); and the upper bound of the Wasserstein distance (f)}\label{fig: Centralized}
\end{figure*} 

Initially, the robots are located at the given initial positions $x_0$, represented by the blue cross marks in Fig. \ref{fig: Centralized}. 
From Algorithm \ref{algorithm:2}, each agent generates a circle centered at the current robot position with a radius of $r$ that gradually increases by the fixed increment $\delta$ until $h$ numbers of sample points are within the circle. 
Once found, all possible trajectories are generated connecting the sample points within the circle starting from the current robot position and the cost for each trajectory is calculated. Next, the trajectory with the minimum cost is selected and the first point of that trajectory is considered as the next goal point $^gx_{t+1}$. The next position of the robot is updated using \eqref{eqn: robot motion control}. Then, the weights for the sample points are updated according to the method explained in the weight update stage. Individual weight information is collected by the supervisory agent and shared with all other agents for the centralized scenario. This procedure continued until the current time step becomes $t = t_e$.

Each agent starting from the given initial position approaches the closest modal Gaussian as shown in Fig \ref{fig: Centralized} (a). 
As the top-left and top-right Gaussian components are covered by two different groups of two agents, these areas are explored faster than the top-middle one covered by a single agent. These two agents finished exploring corresponding areas and then, moved to the top-middle region. After exploring this region, all agents move towards the bottom region that is yet explored as illustrated in Fig. \ref{fig: Centralized} (e).

The upper bound of the Wasserstein distance, $W_{UB}$, is presented in Fig. \ref{fig: Centralized} (f) to provide the performance measure. As five agents formed the ensemble resembles the reference PDF over time, $W_{UB}$ decreased gradually.
At the end of the simulation $t=1000$, the distribution from the five-agent system is very close to the given spatial distribution, resulting in the final value of $W_{UB}$ by $0.2421$, which is really small compare to the initial value of $3600$.

\subsection{Performance Comparison}
The simulation results provided in this part are mainly for the performance comparison between the proposed and Spectral Multiscale Coverage (SMC) methods, proposed in \cite{GM-IM:11}. 
In particular, two comparison subjects are the target detection rate and computation time. 
As the SMC method was developed only for the centralized control strategy for multi-robot explorations, it is logical to compare the performance with the proposed centralized exploration method. 

The multi-robot exploration results for the proposed and SMC methods are presented in Fig. \ref{fig: Exploration Comparison}.
The simulations were conducted with the same parameters in the centralized case, which were used for both the proposed and SMC methods for a fair comparison.
The SMC method utilizes the Fourier basis functions to achieve efficient multi-robot explorations as a tool. Theoretically, infinite numbers of Fourier basis functions are required, which is infeasible in implementation and hence, truncation is necessary. 
The number of Fourier basis functions used for the simulation in Fig. \ref{fig: Exploration Comparison} is $K^2 = 225$ (the square term indicates the two-dimensional case) and the robot positions are updated using the first-order dynamics with the control law from the SMC method.

The proposed OT-based method ended up with a total of $273$ target detection out of $300$ targets (detection rate: $91.33\%$), whereas the SMC method detected a total of $243$ targets (detection rate: $81\%$) after $1000$ time steps. One of the reasons for this result is that the SMC method  detoured the path from one region to another, instead of taking the shortest path, as shown in Fig. \ref{fig: Exploration Comparison} (a), leading to exploration inefficiency. On the other hand, it can be observed in Fig. \ref{fig: Exploration Comparison} (b) that the agents took short paths, enabling them to spend more time in the areas of interest and hence, detect more targets. 

\begin{figure*}[tbph!]
\centering
\subfloat[SMC method]{\includegraphics[scale=0.38]{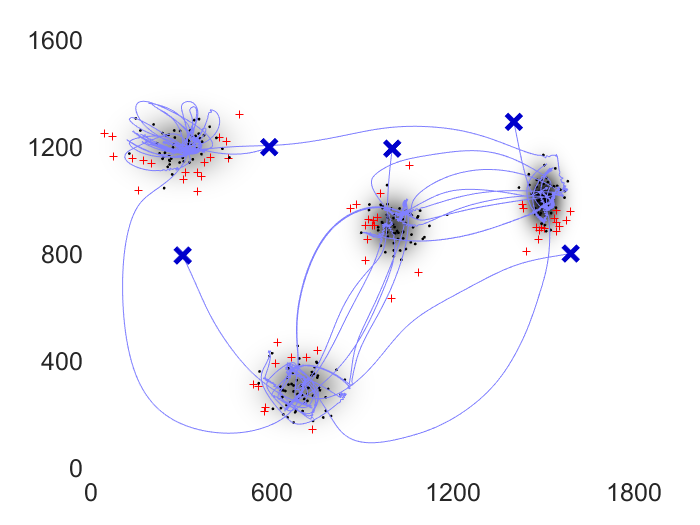}}\qquad
\subfloat[OT-based method]{\includegraphics[scale=0.38]{cen_10.png}}
\caption{Multi-agent exploration trajectories for (a) the SMC method; and (b) the proposed OT-based method}\label{fig: Exploration Comparison}
\end{figure*}

For a better comparison between the OT-based and SMC methods with quantitative results, a total of 50 simulations were conducted with the following conditions:
\begin{itemize}
    \item Total cases: $4$ (OT, SMC with $K = 10, 15,$ and $20$)
    \item Initial robot positions: randomly distributed with uniform distribution in the domain
    \item Target positions: randomly generated with the given reference PDF above
    \item Total simulation runs: 50 (in each run, initial robot positions and target points are randomly generated)
\end{itemize}

\begin{figure*}[tbph!]
\centering
\subfloat[]{\includegraphics[scale=0.35]{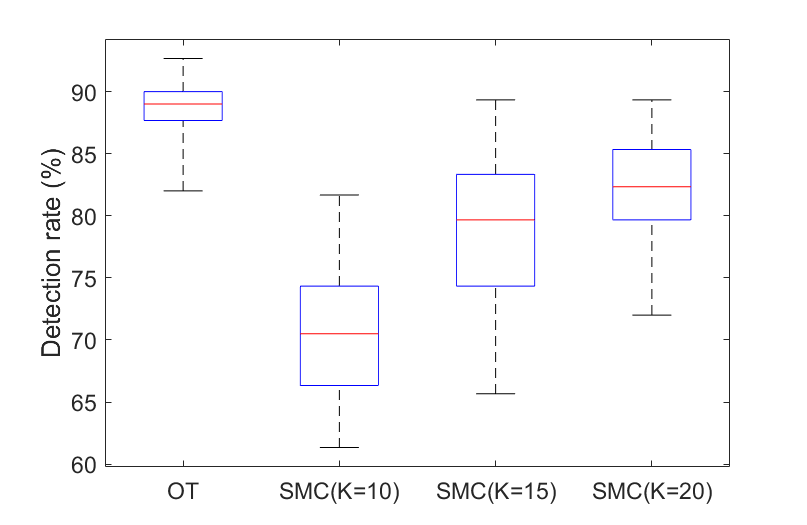}}\qquad
\subfloat[]{\includegraphics[scale=0.35]{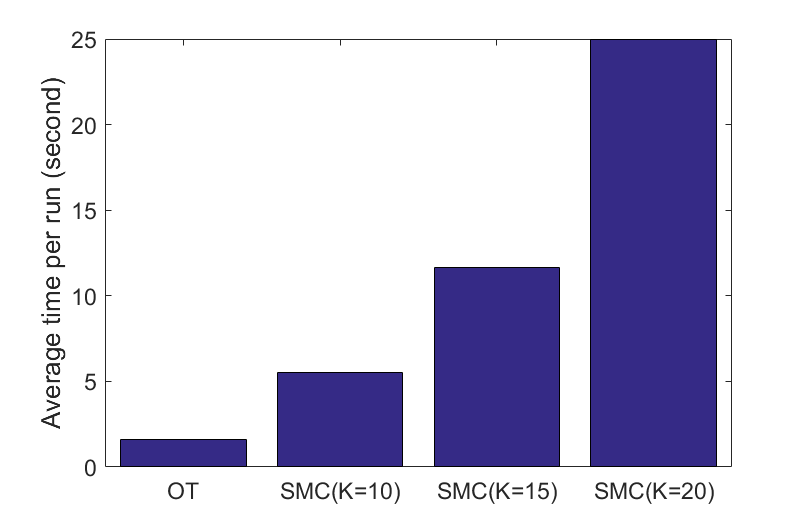}}
\caption{Statistical data for the performance comparison: (a) a boxplot for the target detection rates; and (b) a barplot for the averaged execution time}\label{fig: Statistical Comparison}
\end{figure*}

In Fig. \ref{fig: Statistical Comparison}, the statistical data obtained from 50 simulation runs are presented for the performance comparison between the proposed and SMC methods. From Fig. \ref{fig: Statistical Comparison} (a), the OT-based method outperforms the SMC method in terms of the detection rate (the median detection rate for the OT-based method is 89\%, whereas that for all SMC methods is lower than 82\%). Although increasing $K$ values improves the detection rate of the SMC method, the detection rate for $K = 20$ (which corresponds to $400$ Fourier basis functions for 2D case) is still lower than that for the OT-based method. Moreover, it is observed that the proposed OT-based method is more robust to the randomness of initial robot positions as well as target positions since the boxplot variance of the OT-based method is much smaller than all the other cases. This implies that the SMC method is more sensitive to the given initial conditions.
Fig. \ref{fig: Statistical Comparison} (b) presents the barplots for the averaged execution time. While increasing $K$ leads to a better detection rate, it causes more computation time as a tradeoff.
From both comparison results, it is demonstrated that the proposed OT-based method not only performs better but also takes much less time than the SMC method.

\subsection{Decentralized Multi-Agent Case}
To test the proposed decentralized multi-agent exploration scheme, another simulation was carried out and the results are presented in Fig. \ref{fig: decentralized}. For the decentralized case, the ensemble of the reference PDF is represented by the green dots instead of the colormap. The target detection test is not conducted, however, the upper bound of the Wasserstein distance is provided as the performance measure. 
The spatial distribution is given as a mixture of Gaussian having three modal Gaussian components as follows:
\begin{align*}
    \mu_1 &=[300, 700]^{T},  \mu_2=[1200, 900]^T, \mu_3 = [700, 250]^T\\
    \Sigma_1 &= \begin{bmatrix}
       8000 & 0 \\ 
          0 & 4800
    \end{bmatrix}, 
    \Sigma_2 = \begin{bmatrix}
        3200 & 0\\ 
          0 & 4800
    \end{bmatrix},\\ 
    \Sigma_3 &= \begin{bmatrix}
        6000 & 0 \\ 
           0 & 4800
    \end{bmatrix}
\end{align*}

Other simulation parameters are:
\begin{itemize}
  \item Domain size: $1500  \times 1200$
  \item Number of agents: $n_a = 2$
  \item Maximum allowable number of robot steps: $M=2000$
  \item Number of sample points for the multi-modal Gaussian distribution: $N = 1200$
  \item Initial robot positions:\\
  $x_0 = [1000, 200]^{T},\, [400, 1000]^{T}$
  \item Maximum velocity of the robot: $100$
  \item Robot communication range: $r_{\text{comm.}}=100$
\end{itemize}

\begin{figure*}[tbph!]
\centering
\subfloat[$t=0$]{\includegraphics[scale=0.29]{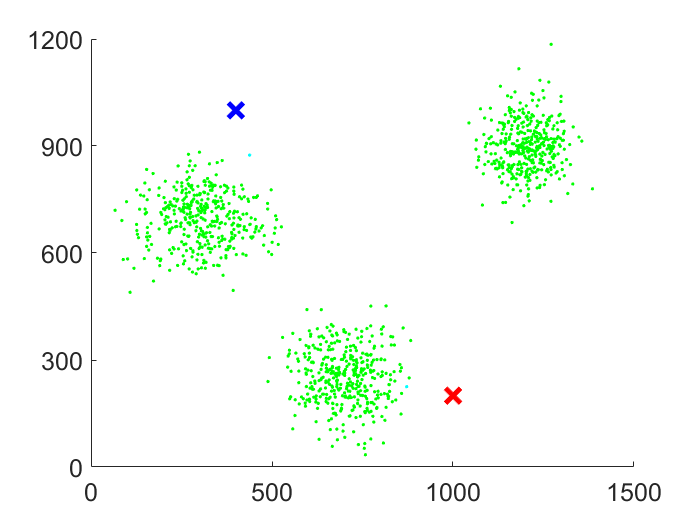}}
\subfloat[$t=300$]{\includegraphics[scale=0.29]{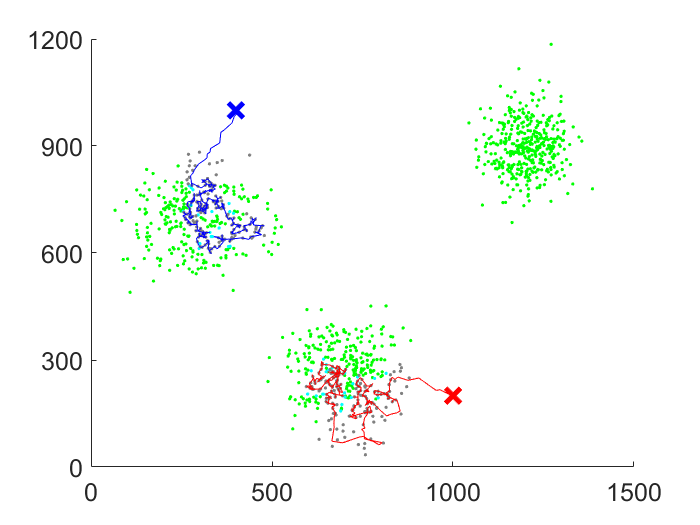}}
\subfloat[$t=600$]{\includegraphics[scale=0.29]{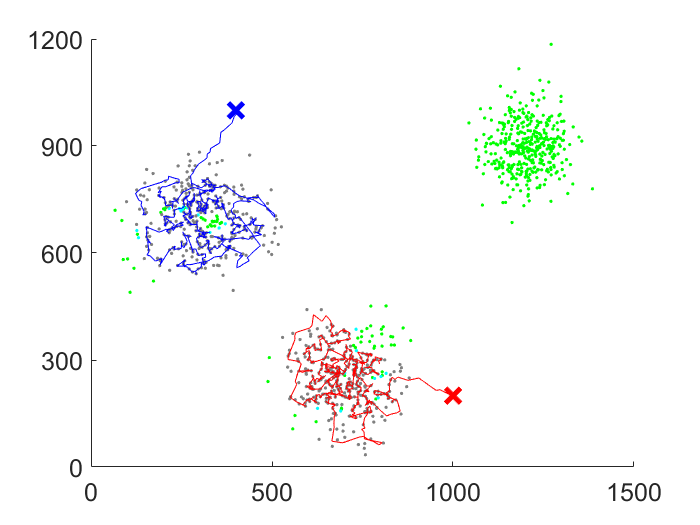}}\\
\subfloat[$t=700$]{\includegraphics[scale=0.29]{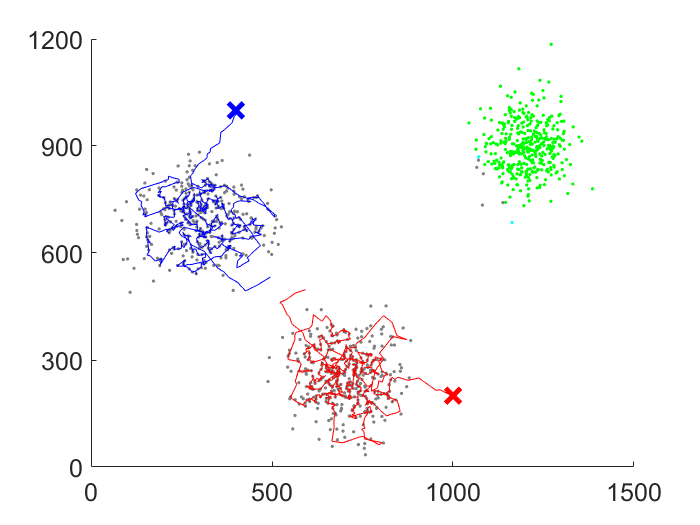}}
\subfloat[$t=1057$]{\includegraphics[scale=0.29]{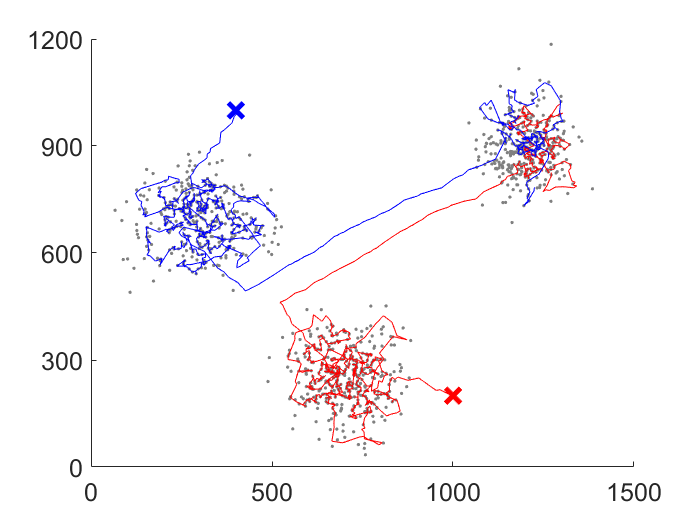}}
\subfloat[$W_{UB}$]{\includegraphics[scale=0.29]{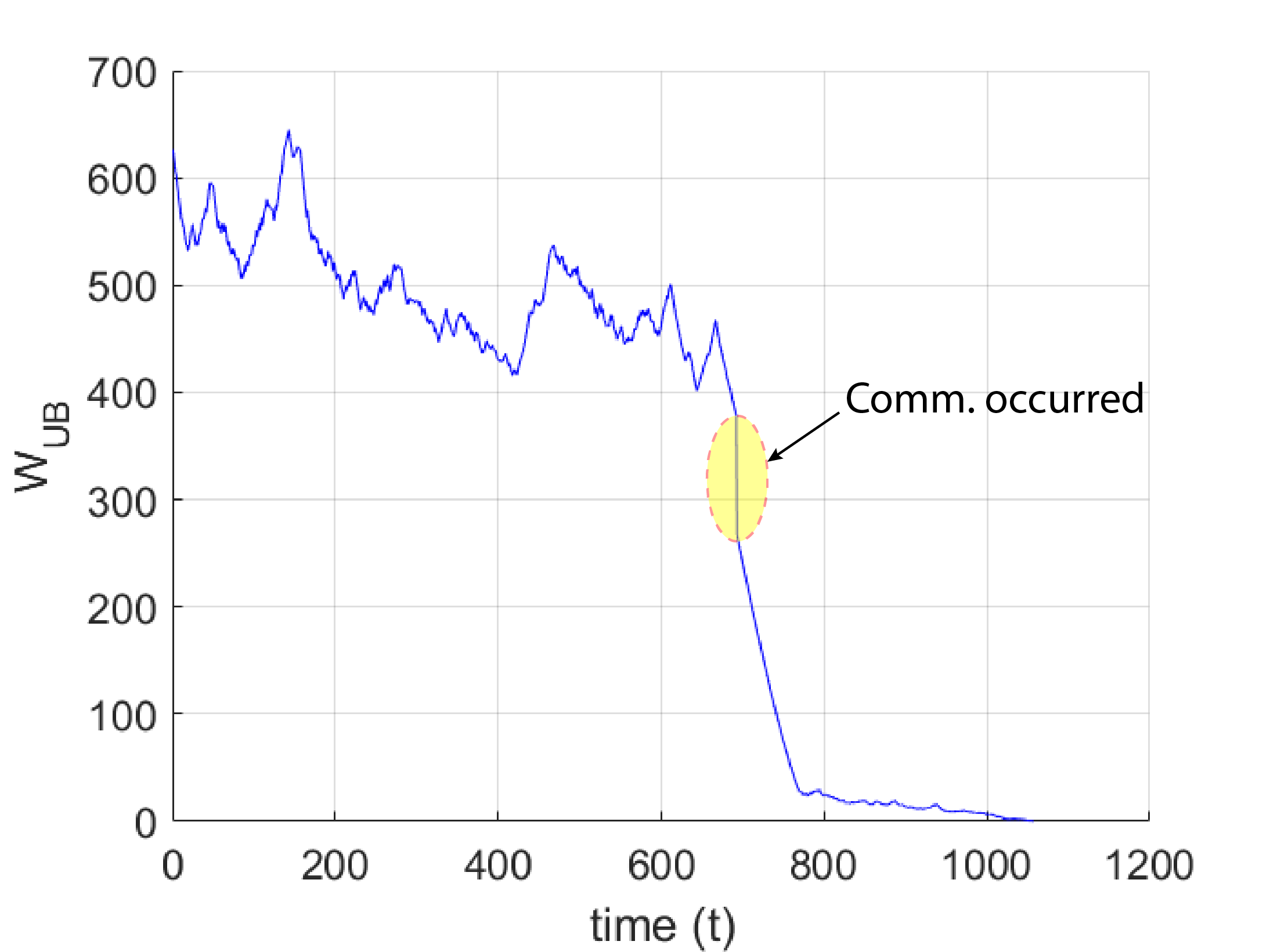}}
\caption{Snapshots of decentralized multi-agent exploration of the given spatial distribution (a)-(e); and the upper bound of the Wasserstein distance (f)}\label{fig: decentralized}
\end{figure*}

In Fig. \ref{fig: decentralized} (a), the initial positions of two agents are indicated by the blue and red crosses and the sample points are presented in green dots with evenly distributed weight initially. As robots explore the domain, the weights for the sample points are depleted, which are described by gray dots.
In Fig. \ref{fig: decentralized} (a), two agents are positioned considerably far away from each other and hence, they are not able to communicate with each other. This resulted in a completely independent exploration plan due to the lack of communication and information sharing. 


Once the exploration of their respective areas is finished (Fig. \ref{fig: decentralized} (d)), each agent traverses the domain to visit the next closest unexplored region. 
When two agents are close enough within the communication range, they start to share and exchange the weight information. As a result, both agents recognize which region is already covered by another agent.
This process is captured in Fig. \ref{fig: decentralized} (e) as they approached each other and then, shared the weight information.
The agents changed the direction and started to travel towards the unexplored region in the top-right. 
 
In Fig. \ref{fig: decentralized} (e), the agents stay within the communication range most of the time while exploring the region, resulting in the efficient exploration by two agents.
From the last figure, it can be observed that the trajectories of two agents do not overlap in most cases, which acts similarly to the centralized exploration strategy. 
The simulation termination time is set up as the largest time among each agent's time to completely deplete the weight of sample points. In this simulation, the exploration has finished at $t = 1057$, which is greater than the centralized effective exploration time $t_e = 1000$, but less than the maximum robot steps $M = 2000$. This implies that the two-agent system effectively explored the domain even with the decentralized control scheme by reducing the time to cover the given areas of interest almost by a half ($1057/2000$).

The upper bound of the Wasserstein distance, $W_{UB}$, is also provided in Fig. \ref{fig: decentralized} (f). While two agents independently cover each region (the left and bottom region, respectively), $W_{UB}$ decreases slowly. This upper bound $W_{UB}$ starts to decrease sharply during the interval (between $t=667$ and $692$) that the red agent traveled toward the left region after completely covering the bottom region. When the two agents encountered with each other ($t=693$), the information exchange has occurred, resulting in a sudden drop of $W_{UB}$ at this instance and then, both agents headed to the right region. In the final step $t=1057$, $W_{UB}$ reached a very small value ($0.1847$) as the agents completed the OT-based exploration. This quantified value using the Wasserstein distance assures that the decentralized multi-robot system attained the efficient exploration.

\subsection{Time-Varying Distribution Case}

For the validation of the adapted strategy on time-varying distributions, simulations were carried out for two different scenarios: time-varying and time-invariant spatial distributions.
For the time-invariant scenario, the targets are moving according to the given random walk model, however, the reference PDF is given as stationary. Therefore, the agents are not able to predict how targets are moving. On the other hand, the spatio-temporal evolution of the reference PDF is incorporated into the exploration plan in the time-varying scenario.
These two different scenarios are considered for the simulation in order to compare the effectiveness of the time-varying distribution with respect to the time-invariant case.

For both scenarios, a bimodal Gaussian distribution is considered as a given spatial distribution. The following parameters are used throughout the simulations for both cases.
\begin{align*}
    \mu_1 &=[600, 600]^{T}, \quad \mu_2=[-50, 0]^T\\
    \Sigma_1 &= \begin{bmatrix}
       40 & 0 \\ 
        0 & 24
    \end{bmatrix}, \quad
    \Sigma_2 = \begin{bmatrix}
        320 & 0\\ 
          0 & 480
    \end{bmatrix}
\end{align*}

The simulation results for the time-invariant and time-varying distributions are provided in Fig. \ref{fig: time-varying} (a)-(c) and (d)-(f), respectively. The initial robot positions are shown in blue crosses, and the detected and undetected targets, respectively, are shown in red pluses and black dots.

The random walk dynamics  \eqref{eqn: animal position update} was applied to the targets, where the diffusion rate constant $v = 7$. 
To make clear presentation of the simulation results, targets are set to be stationary once detected.
The centralized exploration strategy is applied to both cases and the common parameters for the simulations are as follows:

\begin{itemize}
  \item Domain size: $2000 \times 2000$
  \item Number of agents: $n_a = 2$
  \item Maximum number of robot steps: $M=2000$
  \item Effective time steps: $t_e = 1000$
  \item Number of sample points for the bi-modal Gaussian distribution: $N = 1000$
  \item Number of targets: $N_h = 500$
  \item Initial robot positions:\\
  $x_0 = [0, 100]^{T},\, [100, -50]^{T}$
  \item Maximum velocity of the robot: $100$
  \item Robot sensor range: $r_{\text{sensing}} = 15$
\end{itemize}

\begin{figure*}[tbph!]
\centering
\subfloat[$t=0$]{\includegraphics[scale=0.29]{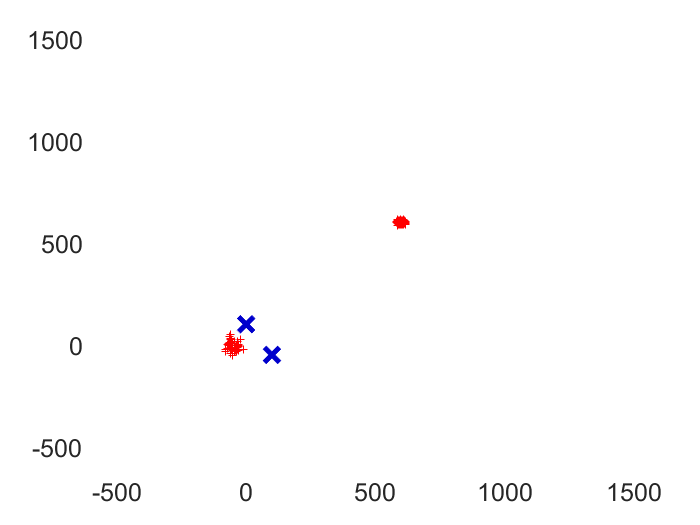}}
\subfloat[$t=500$]{\includegraphics[scale=0.29]{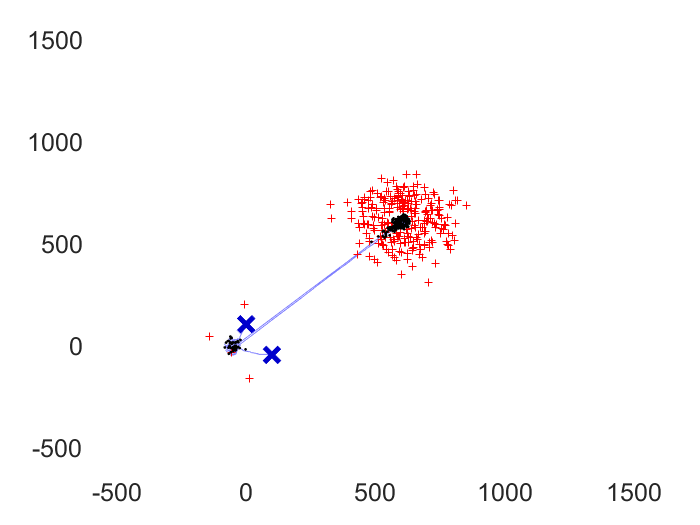}}
\subfloat[$t=1000$]{\includegraphics[scale=0.29]{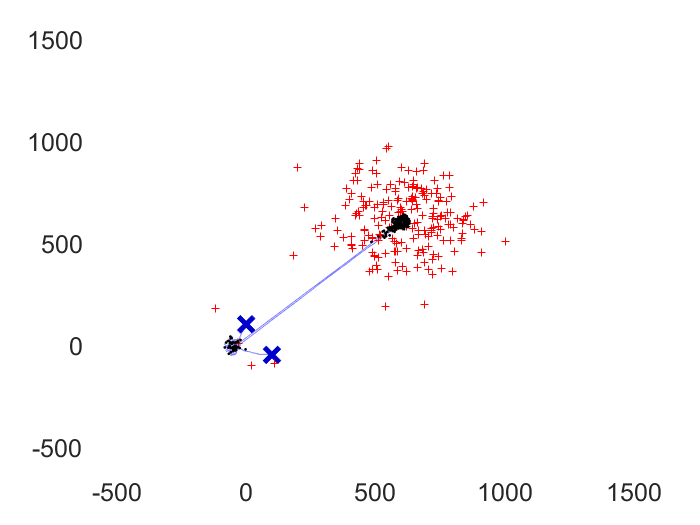}}\\
\subfloat[$t=0$]{\includegraphics[scale=0.29]{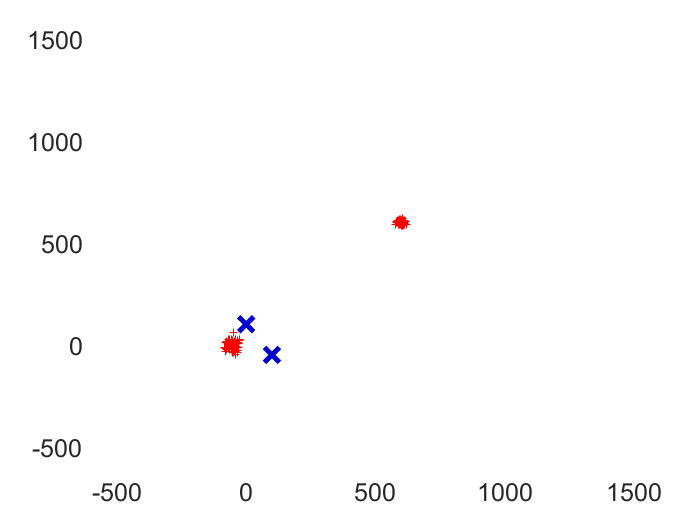}}
\subfloat[$t=500$]{\includegraphics[scale=0.29]{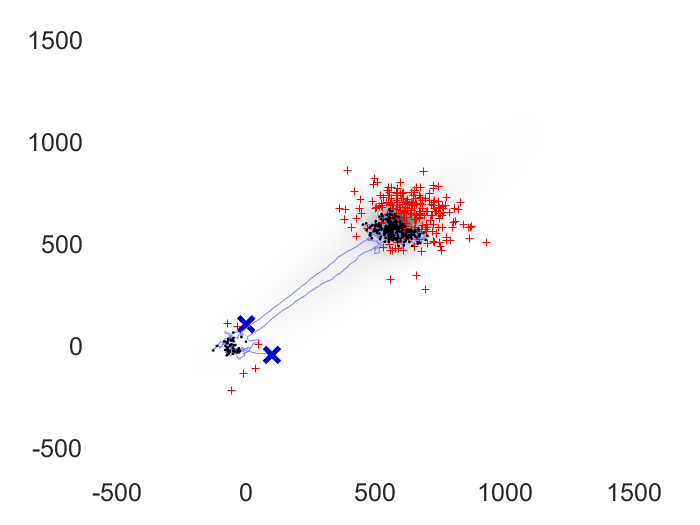}}
\subfloat[$t=1000$]{\includegraphics[scale=0.29]{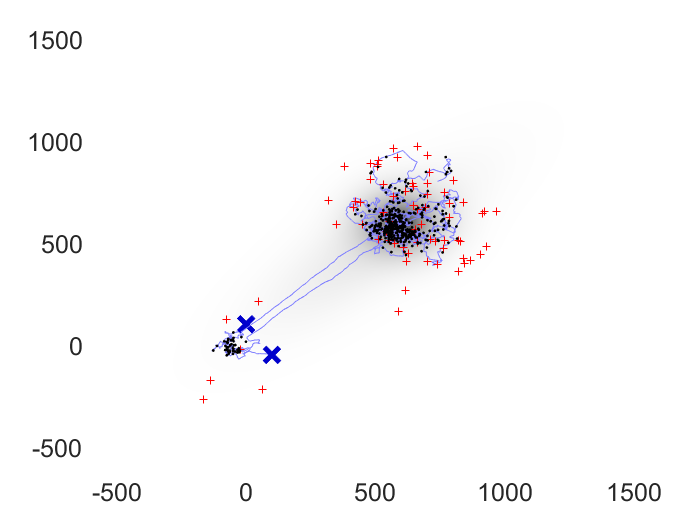}}

\caption{Snapshots of centralized multi-agent explorations of the time-invariant spatial distribution (a)-(c); and time-varying spatial distribution (d)-(f) }\label{fig: time-varying}
\end{figure*} 

For the time-invariant scenario in Fig. \ref{fig: time-varying} (a), both agents are closer to the bottom-left distribution and hence, they visit this region first. Although the sample point representation of the given spatial distribution does not change with time, the target distribution is spreading due to the random nature of target movements.
As the target movement speed is not fast enough, the majority of the targets, located in the bottom-left region, are detected by two agents as delineated in Fig. \ref{fig: time-varying} (b). 
However, the targets in the top-right region had enough time to spread out from the initial distribution when the two agents arrived at this region as illustrated in Fig. \ref{fig: time-varying} (c). Thus, most of the targets in this region were not detected by the agents.

To improve the performance of the exploration strategy, time-varying spatial distribution is considered for the second scenario. As discussed in Section \ref{sec: time-varying},  the sample points representing the spatial distribution are not stationary for this case, rather they also move in a random manner with the given random walk model. During the exploration, the sample points are updated at every time step using \eqref{eqn: sample point update}. As a result, the targets and spatial distributions diffuse at a similar rate.
Figs. \ref{fig: time-varying} (e) and (f) show that the agents detected more targets in the bottom-left as well as top-right distributions as compared to the time-invariant distribution case because of the capability of the proposed scheme to capture the spatio-temporal evolution of the reference PDF.

Fig. \ref{fig: error_bar for time-varying case} provides the detection rate comparison between the time-invariant and time-varying cases with a different diffusion rate constant for both cases.
A total of 10 simulations were conducted in each case as the random walk model led to a stochastic process in both sample and target points movement. 
The vertical line in each point of Fig. \ref{fig: error_bar for time-varying case} indicates a 95\% confidence interval. 
For all diffusion rate values, the time-varying case predominated in the target detection rate. This plot clearly shows that the proposed multi-robot exploration strategy is adaptable to time-varying scenarios and can be employed to many applications where the spatio-temporal evolution of the given distribution needs to be considered.


 

\begin{figure}[!h]
    \centering
    \includegraphics[scale=0.28]{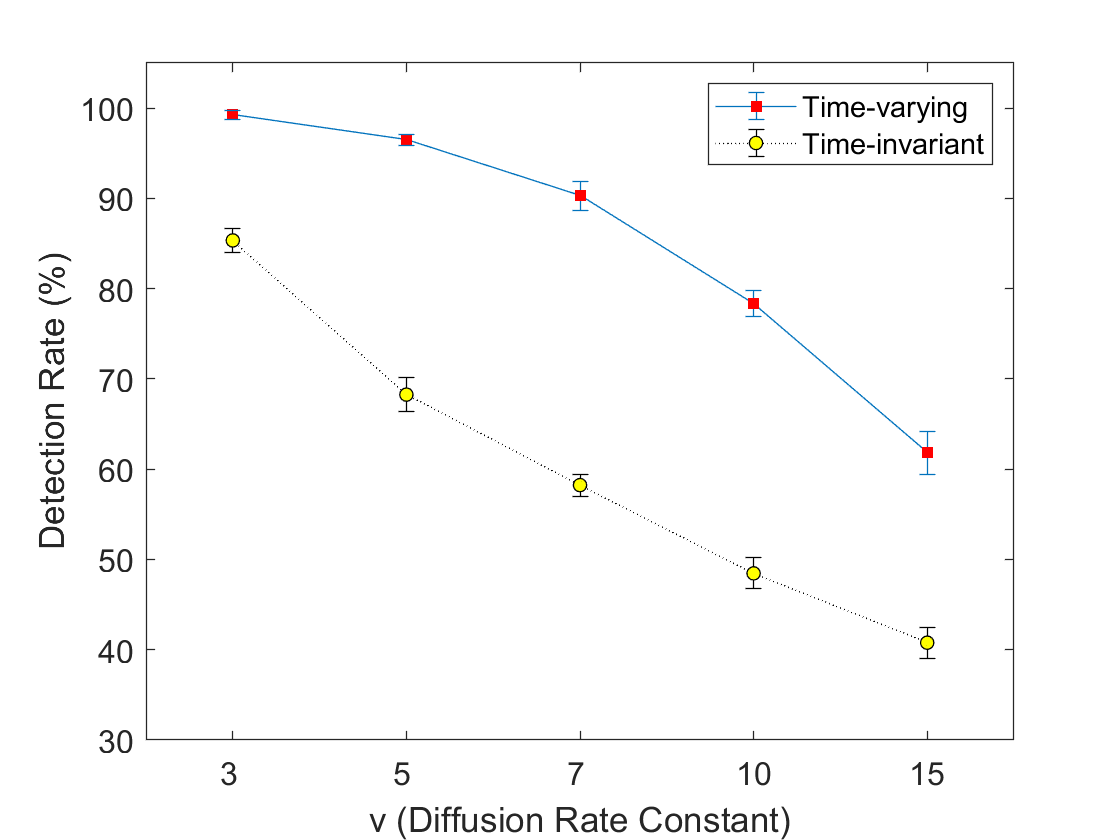}
    \caption{Error bar plot for the detection rate comparison between the time-invariant and time-varying cases; the vertical lines indicate 95\% confidence interval for 10 simulation data in each case}
    \label{fig: error_bar for time-varying case}
\end{figure}

\section{Conclusion}\label{sec: conclusion}
In this paper, we developed an efficient multi-robot exploration scheme, considering the robot energy constraint based on the optimal transport theory. In practice, all agents have finite energy, which needs to be reflected in the exploration plan. The total number of robot points is regarded as finite energy and hence, it can be incorporated into the OT-based multi-robot exploration scheme. Moreover, the proposed methods generate multi-robot trajectories while robot dynamics are separated from it. As a result, the proposed scheme enables a team of heterogeneous robots to be employed in various exploration missions, which broadens the applicability of the developed method. To quantify the performance of the multi-robot exploration scheme, the upper bound of the Wasserstein distance is developed, which can be used to monitor the efficiency of multi-robot explorations in real time. 
Both centralized and decentralized multi-robot exploration plans are proposed with the upper bound of the Wasserstein distance.
Various simulation results are presented for the centralized, decentralized, and time-varying scenarios to verify the effectiveness of the proposed methods.



\bibliographystyle{ieeetr}
\bibliography{OT_multi_robot_exp_journal}

\end{document}